\documentclass[11pt]{article}
\usepackage{geometry}
\usepackage{amsmath,amsfonts}
\usepackage[dvipsnames]{xcolor}
\usepackage{empheq,tensor,cancel,tcolorbox,blkarray,authblk,cite}
\usepackage[colorlinks=true, linkcolor=blue, citecolor=blue, linktoc=all]{hyperref}
\usepackage{bm}
\usepackage{tabularx}
\usepackage{ltablex}
\usepackage{subfigure}

\usepackage{tikz}
\usepackage{cleveref}
\usepackage{supertabular}
\usepackage{todonotes}
\usepackage{mathrsfs}
\usepackage{array}
\usepackage{lipsum}
\usepackage{physics}
\usepackage{tikz-cd}
\usepackage{enumerate}
\usepackage{pifont}
\usepackage{float} 
\usepackage{helvet} 
\usepackage{amsthm} 
\usepackage{amssymb} 
\usepackage{booktabs} 
\usepackage{spectralsequences} 
\usepackage{enumitem}
\setlist{leftmargin=\parindent}

\newcommand{\newreptheorem}[2]{%
\newenvironment{rep#1}[1]{%
 \def\rep@title{#2 \ref{##1}}%
 \begin{rep@theorem}}%
 {\end{rep@theorem}}}
\makeatother
\newtheorem{lemma}{Lemma}[section]
\newreptheorem{lemma}{Lemma}
\newtheorem{theorem}[lemma]{Theorem}

\newreptheorem{conj}{Conjecture}

\usepackage{lscape}
\usepackage{braket}

\usepackage[T1]{fontenc} 

\newcommand{\im}{\mathrm{Im}\,}

\newcommand{\I}{\mathrm{i}}
\newcommand{\E}{\mathrm{e}}

\newcommand{\Z}{{\mathbb{Z}}} 
\newcommand{\R}{{\mathbb{R}}}
\newcommand{\bC}{{\mathbb{C}}}

\newcommand{\td}{{\rm d}}

\newcommand{\be}{\begin{equation}}
\newcommand{\ee}{\end{equation}}
\newcommand{\ba}{\begin{aligned}}
\newcommand{\ea}{\end{aligned}}
\newcommand{\bea}{\begin{eqnarray}}
\newcommand{\eea}{\end{eqnarray}}

\def\mb{\mathbb}

\def\mc{\mathcal}
\def\bp{\begin{pmatrix}}
\def\ep{\end{pmatrix}}

\def\im{\mathrm{im}}
\def\ker{\mathrm{ker}}

\allowdisplaybreaks[1]

\title{\huge On Quantum Aspects of 1-Form Symmetries II: Bordism, Invertible Phases, and Anomalies}

\author[1]{Weizhen Jia\thanks{weizhenjia@cuhk.edu.hk}}
\author[2,3]{Yi-Nan Wang\thanks{ynwang@pku.edu.cn}}
\author[4]{Yi Zhang\thanks{yi.zhang@ipmu.jp}}

\affil[1]{Department of Physics, The Chinese University of Hong Kong, Sha Tin, Hong Kong, China}
\affil[2]{School of Physics, Peking University, Beijing, China, 100871}
\affil[3]{Center for High Energy Physics, Peking University, Beijing, China, 100871}
\affil[4]{Kavli IPMU (WPI), UTIAS, The University of Tokyo, Kashiwa, Chiba 277-8583, Japan}

\date{}

\begin{document}
\maketitle

\begin{abstract}
We study quantum anomalies associated with $U(1)$ 1-form symmetries from the perspective of invertible phases and bordism. We compute the oriented and spin bordism groups of the Eilenberg--Mac Lane space $K(\Z,3)$ up to degree 8 using the Atiyah--Hirzebruch spectral sequence, resolving the relevant extension problems by geometric arguments and identifying both bordism invariants and geometric generators. We then relate these invariants to perturbative and global anomalies, and discuss physical examples and top-down constructions of the corresponding anomaly terms. For 5-dimensional theories, we find a new mixed perturbative anomaly between the $U(1)$ 1-form symmetry and spacetime diffeomorphisms, while for 7-dimensional theories we find a new $\Z_2$-valued discrete anomaly intrinsic to the $U(1)$ 1-form symmetry. We also discuss their boundary realizations and give new physical interpretations of these anomalies.
\end{abstract}
~~\\
\begin{center}
\it --- Dedicated to the memory of Professor Rob Leigh ---
\end{center}

\newpage
\begingroup
\hypersetup{linkcolor=black}
\tableofcontents
\endgroup

\newpage

\section{Introduction and Summary}

Global symmetries are among the most useful structures in quantum field theory. In recent years, the notion of symmetry has been substantially generalized: instead of acting only on local operators, a generalized global symmetry may act on extended operators. A $p$-form global symmetry acts on $p$-dimensional charged objects and couples to a background $(p+1)$-form gauge field~\cite{Gaiotto:2014kfa}. Such symmetries have played an important role in the study of phases of gauge theories, confinement, dualities, higher-group structures, and quantum anomalies; see, for example, recent lecture notes and reviews~\cite{Sharpe:2015mja,Gomes:2023ahz,Schafer-Nameki:2023jdn,Brennan:2023mmt,Bhardwaj:2023kri,Luo:2023ive,Iqbal:2024pee,Costa:2024wks,Kaidi:2026urc}.

In this work we study anomalies associated with $U(1)$ 1-form symmetries. The charged objects are of dimension one, and the corresponding background gauge field is a $U(1)$-gerbe connection, or equivalently a 2-form $U(1)$ gauge field. Topologically, such backgrounds are classified by the Eilenberg--Mac Lane space
\begin{equation}
    B^2U(1)=K(\Z,3)\,.
\end{equation}
Compared with ordinary 0-form symmetries, and also with finite higher-form symmetries, continuous $U(1)$ 1-form symmetries have been less systematically studied from the bordism and invertible field theory perspective. Our goal is to analyze their perturbative and global anomalies using this framework.

The modern viewpoint~\cite{Freed:2014iua} is that an anomaly of a $d$-dimensional QFT is described by an invertible field theory~\cite{Freed:2012bs} in one higher dimension. Here \emph{invertible} means invertible under stacking. On a closed $(d+1)$-manifold, its partition function is a complex number of absolute value one, giving the anomalous phase factor. This is the modern formulation of anomaly inflow, whose classic realization is the Callan--Harvey mechanism~\cite{Callan:1984sa}. Following common terminology, we will use the terms anomaly theory and invertible phase interchangeably. In the condensed matter context, invertible phases protected by a specified symmetry are often called symmetry protected topological (SPT) phases~\cite{Gu:2009dr,Chen:2011pg,Chen:2012ctz,Kapustin:2014tfa,Kitaev:2013}: they become trivial after forgetting the protecting symmetry, while their anomalous boundaries realize the same anomaly inflow picture.

Generally speaking, the anomaly theory need not be topological: it may also depend on geometric data, such as metrics and background connections. The full geometric object contains more information than its deformation class.\footnote{A precise classification of invertible phases requires a differential refinement of the Anderson dual~\cite{Yamashita:2021cao}.} In this paper we focus on these deformation classes, equivalently on the anomaly data at the bordism level. In the Freed--Hopkins framework \cite{Freed:2016rqq}, these classes are captured by the Anderson dual of the corresponding bordism theory. Thus, for a spacetime structure $\mathcal S$ and a target space $X$ classifying the background fields, the relevant group is
\begin{equation}
    (I_{\Z}\Omega^{\mathcal S})^{d+2}(X)\,,
\end{equation}
which is isomorphic (not canonically) to
\begin{equation}
{\rm Hom}_{\Z}\!\left({\rm Tor}\,\Omega^{\mathcal S}_{d+1}(X),\R/\Z\right) \oplus     \mathrm{Hom}_{\Z}\!\left({\rm Free}\,\Omega^{\mathcal S}_{d+2}(X),\Z\right)\,.
\end{equation}
The shift by two is the standard degree convention for the Anderson dual of the bordism spectrum. The free part is represented by local anomaly polynomials, which can be described by the BRST cohomology and descent formalism. In \cite{PartI} we used the geometric viewpoint of~\cite{Jia:2023tki} and extended it to gerbe backgrounds for $U(1)$ 1-form symmetries. The torsion part is invisible to local anomaly polynomials and gives global, and hence non-perturbative, anomalies. Bordism groups therefore organize both perturbative and global anomalies in a single framework. For $U(1)$ 1-form symmetries, we take $X=K(\Z,3)$ and compute the oriented and spin bordism groups
\begin{equation}
    \Omega^{\rm SO}_\bullet(K(\Z,3))\qquad\text{and}
    \qquad
    \Omega^{\rm Spin}_\bullet(K(\Z,3))
\end{equation}
up to degree 8. Physically, they correspond to bosonic and fermionic theories without time-reversal symmetry, respectively. The results are summarized in Table~\ref{tab:bordism-KZ3-summary}.
\begin{table}[H]
\centering
    \begin{tabular}{c|ccccccccc}
  $n$ & $0$ & $1$ & $2$ & $3$ & $4$ & $5$ & $6$ & $7$ & $8$  \\ \midrule
$\Omega^{\rm SO}_n(K(\Z,3))$
& $\Z$ & $0$ & $0$ & $\Z$ & $\Z$ & $\Z_2^2$ & $0$ & $\Z$ & $\Z_2^2\oplus \Z^2$   \\  \midrule
$\Omega^{\rm Spin}_n(K(\Z,3))$
& $\Z$ & $\Z_2$ & $\Z_2$ & $\Z$ & $\Z$ & $0$ & $0$ & $\Z$ & $\Z_2\oplus \Z^2$
\end{tabular}
\caption{Results for $\Omega^{{\rm SO}/\rm Spin}_{\bullet\leqslant 8}(K(\Z,3))$.}
\label{tab:bordism-KZ3-summary}
\end{table}
The spin bordism groups agree with the recent computation of~\cite{Joyce:2025jyd} (see also~\cite{Dierigl:2026sok}). We compute the oriented bordism groups using the Atiyah--Hirzebruch spectral sequence (AHSS) together with geometric arguments that resolve the relevant extension problems. In particular, the degree-7 oriented bordism group involves a nontrivial extension, which we solve using an $S^3$-bundle construction related to Milnor's construction of exotic spheres.

With the bordism results at hand, one obtains $(I_{\Z}\Omega^{\mathcal S})^{d+2}(K(\Z,3))$ for $d\leqslant 7$, which we present in Table~\ref{tab:invphase-KZ3-intro}.
\begin{table}[H]
\centering
    \begin{tabular}{c|cccccccccc}
  $n= d+2$ & $0$ & $1$ & $2$ & $3$ & $4$ & $5$ & $6$ & $7$ & $8$ & $9$ \\ \midrule
$(I_{\Z} \Omega^{\rm SO})^{n}(K(\Z,3))$ & $\Z$ & $0$ & $0$ & $\Z$ & $\Z$ & $0$ & $\Z_2^2$ & $ \Z$ & $   \Z^2 $ & $  \Z_2^2  $  \\  \midrule
$(I_{\Z} \Omega^{\rm Spin})^{n}(K(\Z,3))$ & $\Z$ & $0$ & $\Z_2$ & $\Z_2\oplus\Z$ & $\Z$ & $0$ & $0$ & $\Z$ & $  \Z^2$ & $  \Z_2  $
\end{tabular}
\caption{Invertible phases as Anderson duals of the bordism groups.}
\label{tab:invphase-KZ3-intro}
\end{table}

The low-dimensional entries reproduce familiar structures: purely gravitational phases, fermion-parity/spin-structure effects, and known mixed anomalies of Maxwell theory. The new content of this work starts in higher dimensions. As discussed in \cite{PartI}, for $5d$ theories with a $U(1)$ 1-form symmetry, there is a mixed perturbative anomaly between the 1-form symmetry and spacetime diffeomorphism. In the oriented case the anomaly polynomial is generated by
\begin{equation}
    H_3\wedge p_1\,,
\end{equation}
where $H_3$ is the field strength of the background 2-form gauge field. In the spin case the corresponding normalization is $\frac{1}{4}  H_3\wedge p_1$. In this work, we further investigate the physical effect of this anomaly and show that magnetic strings in such a phase carry additional topological data associated with trivializations of the relevant characteristic class near the string worldsheet.

For 7-dimensional theories, we find a new $\Z_2$-valued discrete anomaly intrinsic to the continuous $U(1)$ 1-form symmetry,
\begin{equation}
    u\smile \mathrm{Sq}^2 u\,,
\end{equation}
where $u$ is the mod-$2$ reduction of the universal degree-3 class on $K(\Z,3)$. On oriented non-spin manifolds there is an additional mixed anomaly
\begin{equation}
    u\smile w_2\smile w_3\,.
\end{equation}
We discuss possible boundary interpretations of these anomalies, their relation to anomaly interplay, and their restriction to finite subgroups of $U(1)$.

We also discuss examples involving both electric and magnetic higher-form symmetries in Maxwell-type theories. In $4d$, the bordism framework reproduces the known phases of Maxwell theory associated with its electric and magnetic $U(1)$ 1-form symmetries. This includes, in particular, the mixed electric-magnetic anomaly derived in \cite{PartI} by the BV--BRST descent, with anomaly polynomial $I_6\sim H_E\wedge H_M$, which can be viewed as an ABJ anomaly after gauging one of the two $U(1)$ 1-form symmetries. In $5d$ and $7d$ Maxwell theories, including the magnetic dual symmetry leads to additional mixed anomalies involving the corresponding higher-form background fields.

Finally, we describe top-down constructions of some of these anomaly terms from string theory. In particular, we study reductions of type IIA topological couplings and explain how terms of the form $H_3\wedge p_1$ and $u\,\mathrm{Sq}^2 u$ can arise from higher-dimensional topological actions.

The paper is organized as follows. In Section~\ref{sec:generalities}, we review the traditional and modern viewpoints on anomalies and explain the bordism framework used in the rest of the paper. In Section~\ref{sec:KZ3}, we compute the relevant oriented and spin bordism groups of $K(\Z,3)$. In Section~\ref{sec:physics}, we discuss physical examples of anomaly theories for continuous $U(1)$ 1-form symmetries. In Section~\ref{sec:topdown}, we describe top-down constructions of the corresponding anomalies. We conclude in Section~\ref{sec:conclusion}. The appendices collect the bordism computations for products of Eilenberg--Mac Lane spaces used in the main text.

\section{Generalities about Anomalies in QFT}
\label{sec:generalities}
Traditionally, anomalies in quantum field theory are viewed as obstructions to preserving a classical symmetry after quantization \cite{Schwinger:1951nm,Johnson:1963vz,adler1969axial,bell1969pcac,Fujikawa:1979ay,Fujikawa:1980eg,Harvey:2005it}. In the path integral formulation, this obstruction appears as a non-invariance of the quantum partition function under a transformation which is a symmetry of the classical theory \cite{Fujikawa:1979ay}. Consider a $d$-dimensional QFT on a spacetime manifold $M_d$ with global symmetry group $G$, coupled to a background $G$-connection $A$. For a non-anomalous symmetry, the partition function is invariant under background gauge transformations $g$, i.e.,
\begin{equation}
  Z_{M_d}[A^g]=Z_{M_d}[A]\,.
\end{equation}
In an anomalous theory, this invariance is replaced by multiplication by a phase:
\begin{equation} \label{eq:anomalousphase}
  Z_{M_d}[A^g]
  =
  \exp\!\left(2\pi \I \,\mathcal I(M_d;A,g)\right)
  Z_{M_d}[A]\,.
\end{equation}
Equivalently, the corresponding Ward identity is violated. For an infinitesimal transformation, namely when $g = \E^\lambda$ is infinitesimally close to the identity, the anomalous phase factor gives the perturbative, or local, anomaly. In this case one usually starts from an invariant closed $(d+2)$-form $I_{d+2}$, the anomaly polynomial. Locally one chooses a Chern--Simons form $I^{(0)}_{d+1}$ and applies the descent equations
\begin{align}
    I_{d+2}=\td I^{(0)}_{d+1}\,,
    \qquad
    \delta_{\lambda} I^{(0)}_{d+1}=\td I^{(1)}_d(\lambda)\,,
\end{align}
so that the infinitesimal anomalous phase factor is computed by
\begin{align}
    \delta_{\lambda}\log Z_{M_d}[A]
    =
    2\pi \I \int_{M_d} I^{(1)}_d(\lambda)\,.
\end{align}
The Adler--Bell--Jackiw anomaly \cite{adler1969axial,bell1969pcac} and perturbative gravitational anomalies \cite{Alvarez-Gaume:1983ihn} of chiral fermions are standard examples.

There may also exist anomalous phase factors that are invisible to infinitesimal variations and arise from nontrivial topology in the space of background fields.  These are called global anomalies and are detected by large gauge transformations. Witten's $SU(2)$ anomaly is the prototypical example~\cite{Witten:1982fp}, in this particular case $g$ is the large gauge transformation representing the nontrivial element in $\pi_4(S^3) \cong \Z_2$. 

A more modern viewpoint is that an anomaly is not merely a failure of invariance, but is itself a quantum field theory \cite{Monnier:2014rua,Monnier:2019ytc,Witten:2019bou,Witten:2015aba}. In the formulation of Freed~\cite{Freed:2014iua}, the anomaly of a $d$-dimensional QFT is encoded by an invertible extended $(d+1)$-dimensional field theory
\begin{equation}
  \alpha:
  \mathrm{Bord}_{\langle d-1,d,d+1\rangle}(\mathcal F')
  \longrightarrow
  \Sigma^{d+1}I_{\R/\Z}\,.
\end{equation}
Here $\mathcal F'$ denotes the relevant background fields, such as a metric, orientation, spin structure, gauge bundle, connection, and so on. The domain $\mathrm{Bord}_{\langle d-1,d,d+1\rangle}(\mathcal F')$ is the corresponding bordism category of manifolds equipped with these fields, with bordisms composed by gluing.
The notation $\langle d-1,d,d+1\rangle$ means that the anomaly theory is extended down to codimension two: it assigns compatible data to closed $(d+1)$-, $d$-, and $(d-1)$-dimensional manifolds with background fields. The codomain $\Sigma^{d+1}I_{\R/\Z}$ is Freed's compact notation for the universal target of such invertible anomaly theories. We will not need its detailed construction. The point is that, after exponentiation, $\alpha$ assigns phases to closed $(d+1)$-manifolds and a 1-dimensional complex vector space to closed $d$-manifolds.

Concretely, on a closed $(d+1)$-manifold we write the value of the anomaly theory as
\begin{equation}
    W_{d+1}\longmapsto
    Z_{\alpha}(W_{d+1})
    =
    \exp(2\pi \I \alpha(W_{d+1}))\in U(1)\,,
\end{equation}
and on a closed $d$-manifold it assigns a 1-dimensional complex vector space, called the anomaly line:
\begin{equation}
    M_d\longmapsto \alpha(M_d)\,.
\end{equation}
The partition function of an anomalous theory is naturally an element of this line,
\begin{equation}
    Z(M_d)\in \alpha(M_d)\,,
\end{equation}
rather than a canonically defined complex number. A trivialization of the anomaly line would turn it into an ordinary complex-valued partition function. The anomaly is the obstruction to choosing such trivializations compatibly over background fields and under gluing. In this sense, an anomalous theory is a relative field theory~\cite{Freed:2012bs}, defined relative to its anomaly theory $\alpha$.

The anomalous phase factor in~\eqref{eq:anomalousphase} is one manifestation of the same anomaly theory. The transformation $g$ defines a loop in the space of background fields, or equivalently a $(d+1)$-dimensional mapping torus $T_{d+1}(M_d,A,g)$ obtained from $M_d\times[0,1]$ by identifying the two ends using $g$. With the convention of~\eqref{eq:anomalousphase}, the anomaly theory evaluates on this mapping torus to the phase\footnote{This manifestation is also called Dai--Freed anomaly~\cite{Dai:1994kq,Garcia-Etxebarria:2018ajm}.} 
\begin{equation}
  Z_{\alpha}\!\left(T_{d+1}(M_d,A,g)\right) \in U(1)\,.
\end{equation}
Depending on the nature of the transformation $g$, this holonomy data reproduces the usual distinction between perturbative and global anomalies. For $g$ infinitesimally close to the identity, one obtains the local, perturbative anomaly, described by the curvature or infinitesimal variation of the anomaly theory. For large transformations, or more generally for nontrivial loops in the space of background fields, the same anomaly theory detects global anomalies through its holonomy. In this sense, the modern framework puts perturbative and global anomalies on equal footing.

When a local Chern--Simons representative exists, its gauge variation gives the descent expression above. For example, for a 2-dimensional theory with a 3-dimensional Chern--Simons anomaly theory, the $2d$ anomaly is obtained from $\delta {\rm CS}_3$, while the same information is encoded globally by the mapping torus partition function.

However, one should distinguish the general anomaly field theory from the data used for classification. In Freed's formulation, the anomaly theory $\alpha$ need \emph{not} be topological; it may depend on geometric data such as metrics and connections. This is essential for local gauge and gravitational anomalies, whose differential-geometric representatives involve Chern--Simons forms, eta invariants, determinant lines, Pfaffian lines, and their connections.

For classification purposes, we keep only the deformation class of $\alpha$, or equivalently the bordism invariants that capture the anomaly data. This is the level at which the bordism classification~\cite{Freed:2016rqq,Yonekura:2018ufj} applies; see also~\cite{Lee:2020ewl,Tachikawa:2021mby} for related uses in physics. More precisely, this framework classifies reflection-positive invertible phases with a specified symmetry type. The symmetry type, often denoted by $\mathcal B$, is also called the stable tangential structure in the mathematical literature. It includes both spacetime structures, such as $\rm SO $, $\rm Spin$, $\rm O$, $\rm Pin^\pm$ or string structures, and internal symmetry backgrounds.

Once a symmetry type $\mathcal B$ is specified, the possible anomaly classes in $(d+1)$ dimensions form the group
\begin{equation}
    {\rm Inv}_{\mathcal B}^{d+1}
    :=
    (I_{\Z}\Omega^{\mathcal B})^{d+2}({\rm pt}) \,.
\end{equation}
Here $\Omega^{\mathcal B}$ is the bordism theory associated with the symmetry type $\mathcal B$, and $I_{\Z}\Omega^{\mathcal B}$ denotes its Anderson dual. Thus, the formulation above keeps track of the anomaly field theory itself, while $(I_{\Z}\Omega^{\mathcal B})^{d+2}$ is the computable group of deformation classes used below.

It is practical to write the symmetry type $\mathcal B$ as a spacetime structure $\mathcal S$ together with a map from the spacetime manifold to a target space $X$ classifying internal (generalized) symmetry backgrounds. With this convention,
\begin{equation}
    \Omega^{\mathcal B}_d
    =
    \Omega^{\mathcal S}_d(X)
    :=
    \frac{
    \{\text{$d$-dimensional manifolds with $\mathcal S$-structure and a map to $X$}\}
    }{\rm bordism} \, .
\end{equation}
The corresponding group of bordism-level anomaly classes is usually expressed as
\begin{equation}
    (I_{\Z}\Omega^{\mathcal S})^{d+2}(X)\,,
\end{equation}
and it fits into the universal coefficient short exact sequence
\begin{equation} \label{eq:UCTforIOMEGA}
      0
      \longrightarrow
      \mathrm{Ext}^1_{\Z}\!\left(\Omega^{\mathcal S}_{d+1}(X),\Z\right)
      \longrightarrow
      (I_{\Z}\Omega^{\mathcal S})^{d+2}(X)
      \longrightarrow
      \mathrm{Hom}_{\Z}\!\left(\Omega^{\mathcal S}_{d+2}(X),\Z\right)
      \longrightarrow
      0 \, .
\end{equation}
For the finitely generated bordism groups considered here, the first term can be written as
\begin{equation}
    \mathrm{Ext}^1_{\Z}\!\left(\Omega^{\mathcal S}_{d+1}(X),\Z\right)
    =
    {\rm Hom}_{\Z}\!\left({\rm Tor}\,\Omega^{\mathcal S}_{d+1}(X),\R/\Z\right)\,,
\end{equation}
and accounts for the torsion, or global, part of the anomaly. The last term
\begin{equation}
    \mathrm{Hom}_{\Z}\!\left(\Omega^{\mathcal S}_{d+2}(X),\Z\right)
    =
    \mathrm{Hom}_{\Z}\!\left({\rm Free}\,\Omega^{\mathcal S}_{d+2}(X),\Z\right)
\end{equation}
captures the free part, detected by standard anomaly polynomials. Consequently, the computation of bordism-level anomalies reduces to determining the bordism groups $\Omega^{\mathcal S}_\bullet(X)$ and identifying the corresponding bordism invariants that detect their generators.

This framework applies not only to anomalies of ordinary 0-form continuous symmetries, but also to 0-form finite symmetries, higher-form symmetries and higher-group symmetries~\cite{Wan:2019oax,Wan:2019soo,Lee:2020ewl,Davighi:2023luh,Debray:2023rlx}. For a continuous $U(1)$ 1-form symmetry, the background gauge field is a $U(1)$-gerbe connection, or a 2-form $U(1)$ gauge field. Its topological class is classified by $B^2U(1)=K(\Z,3)$. Thus, in this paper we take
$X=K(\Z,3)$,
and study the anomaly groups
\begin{equation}
    (I_{\Z}\Omega^{\mathcal S})^{d+2}(K(\Z,3))\,,
    \qquad
    \mathcal S={\rm SO},\,{\rm Spin}\,.
\end{equation}
That is, we focus on bosonic and fermionic theories without time-reversal symmetry; including time reversal would require replacing these by the appropriate unoriented or pin structures. The bordism groups $\Omega^{\mathcal S}_\bullet(K(\Z,3))$ can be computed using the Atiyah--Hirzebruch spectral sequence (AHSS), one of the standard approaches. In this method, however, determining the differentials and resolving extension problems can be difficult. Another approach is the Adams spectral sequence (ASS), which exploits the full structure of modules over the Steenrod algebra, though its main difficulty lies in the required homological algebra. In this work, we adopt the suspension trick used in \cite{Joyce:2025jyd} and compute $\Omega^{\mathrm{SO}}_{\bullet \leqslant 8}(K(\Z,3))$ using the AHSS together with geometric arguments.

The resulting bordism invariants provide the anomaly terms used in the physical examples below, where the free classes give perturbative anomalies and the torsion classes give global or discrete anomalies.

\section{Bordism Computations for $\Omega^{\mathcal{S}}_{d\leqslant 8}(K(\Z,3))$}

\label{sec:KZ3}
As announced in the introduction, in this section we compute $\Omega^{\mathcal{S}}_{d\leqslant 8}(K(\Z,3))$ for $\mathcal{S} = \rm SO$ or $\rm Spin$, with the results summarized in Table~\ref{tab:bordism-KZ3-summary}. Readers interested only in the results and their physical interpretations may skip directly to Section~\ref{sec:physics}.

Our approach makes use of geometric arguments based on the AHSS (see, for instance~\cite{HatcherSS}, for a review) applied to the trivial fibration  \begin{equation} \label{eq:trivialfib}
    \text{pt} \longrightarrow X \longrightarrow X \,.
\end{equation}
 Since the base $K(\Z,3)$ is simply-connected, its fundamental group, $\pi_1(K(\Z,3)) = 0$, acts trivially on $\Omega^{\mathcal{S}}_\bullet({\rm pt})$, and therefore the AHSS has
\begin{equation} \label{eq:AHSSKZ3}
    E^2_{p,q} = H_p(K(\Z,3),\Omega^{\mathcal{S}}_q({\rm pt})) \Longrightarrow  \Omega^{\mathcal{S}}_{p+q}(K(\Z,3)) \,.
\end{equation}
To determine the $E^2$-page, one needs the bordism groups $\Omega^{\mathcal S}_q(\mathrm{pt})$ of the fiber $\mathrm{pt}$ and the homology groups $H_p(K(\Z,3))$ of the base in the fibration~\eqref{eq:trivialfib}.
 $\Omega^{\rm SO/ Spin}_{d\leqslant 9}(\rm pt)$  can be found in, e.g.,~\cite{MilnorStasheff1974, Garcia-Etxebarria:2018ajm}:
\begin{equation} \label{eq:pointbordism}
    \begin{tabular}{c|cccccccccc}
  $q$ & $0$ & $1$ & $2$ & $3$ & $4$ & $5$ & $6$ & $7$ & $8$ & $9$ \\ \midrule
$\Omega^{\rm SO}_q({\rm pt})$ & $\Z$ & $0$ & $0$ & $0$ & $ \Z $ & $\Z_2$ & $0$ & $0$ & $\Z^2$ & $\Z_2^2$  \\  \midrule
$\Omega^{\rm Spin}_q({\rm pt})$ & $\Z$ & $\Z_2$ & $\Z_2$ & $0$ & $\Z$ & $0$ & $0$ & $0$ & $\Z^2$ & $\Z_2^2$
\end{tabular}
\end{equation}

\paragraph{Classifying space $K(\Z,3)$ and its (co)homologies.} We are interested in anomalies of $U(1)$ 1-form symmetries, whose background gauge field is a $U(1)$-gerbe connection $B_{\mu\nu}$, i.e., a 2-form $U(1)$ gauge field. The classifying space for principal $U(1)$-bundles is the \emph{Eilenberg--Mac Lane} space
\begin{equation}
    BU(1) = K(\Z,2 ) \cong \bC  \mathbb{P}^\infty\,.
\end{equation}
More generally, the higher classifying space for an $n$-form $U(1)$ gauge field is
\begin{equation}
   B^nU(1) = K(\Z, n+1) \,.
\end{equation}
In particular, for a $U(1)$-gerbe we have $B^{2}U(1)=K(\Z,3)$, whose (co)homology groups are well known and are summarized below.
\begin{itemize}
    \item Integral homology group (see~\cite{Breen:2016dfpf} or page 404 of \cite{hatcher2002algebraic})
\begin{equation} \label{eq:integralhomology}
    \begin{tabular}{c|cccccccccc}
  $p$ & $0$ & $1$ & $2$ & $3$ & $4$ & $5$ & $6$ & $7$ & $8$ & $9$ \\ \midrule
$H_p(K(\Z, 3),\Z)$ & $\Z$ & $0$ & $0$ & $\Z$ & $ 0 $ & $\Z_2$ & $0$ & $\Z_3$ & $\Z_2$ & $\Z_2$
\end{tabular}
\end{equation}

    \item Integral cohomology group (see~\cite{Breen:2016dfpf} and Example~1.19 of \cite{HatcherSS})
    \begin{equation} \label{eq:integralcohomology}
    \begin{tabular}{c|cccccccccc}
  $p$ & $0$ & $1$ & $2$ & $3$ & $4$ & $5$ & $6$ & $7$ & $8$ & $9$ \\ \midrule
$H^p(K(\Z, 3),\Z)$ & $\Z$ & $0$ & $0$ & $\Z \,\iota$ & $ 0 $ & $0$ & $\Z_2 \, \iota^2 $ & $0$ & $\Z_3\, \kappa$ & $\Z_2 \,\iota ^3$
\end{tabular}
\end{equation}
    Let $\iota$ be the generator $H^3(K(\Z, 3),\Z)$. Then $H^6(K(\Z, 3),\Z)$ is generated by $ \iota \smile \iota$, also denote by $\iota^2$ for brevity. Graded commutativity of cup product implies that $2 (\iota \smile \iota) =0$, but does not force
$\iota \smile \iota$ itself to vanish; in fact, in this example we have $\iota \smile \iota \neq 0$. This can be understood by computing $H^\bullet(K(\Z, 3),\Z)$ using the standard path fibration and spectral sequence (see Example 5.20 of~\cite{HatcherSS})
\begin{equation}
   \Omega K(\Z,3) \cong K(\Z, 2)\longrightarrow P \longrightarrow K(\Z, 3) \,.
\end{equation}
We list the generators of the integral cohomology of $K(\Z,3)$ up to degree 9 for our further computations and note that there is only one independent generator $\kappa$ emerging.

    \item Mod-2 homology of $K(\Z, 3)$ can be computed using the universal coefficient theorem
\begin{equation} \label{eq:mod2homology}
    \begin{tabular}{c|cccccccccc}
  $p$ & $0$ & $1$ & $2$ & $3$ & $4$ & $5$ & $6$ & $7$ & $8$ & $9$ \\ \midrule
$H_p(K(\Z, 3),\Z_2)$ & $\Z_2$ & $0$ & $0$ & $\Z_2$ & $ 0 $ & $\Z_2$ & $\Z_2$ & $0$ & $\Z_2$ & $\Z_2^2$
\end{tabular}
\end{equation}
    \item Mod-2 cohomology ring~\cite{Cartan1954EM,HatcherSS, MimuraToda1991}
\begin{equation}\label{eq:mod2cohomology}
    \begin{tabular}{c|cccccccccc}
  $p$ & $0$ & $1$ & $2$ & $3$ & $4$ & $5$ & $6$ & $7$ & $8$ & $9$ \\ \midrule
$H^p(K(\Z, 3),\Z_2)$ & $\Z_2$ & $0$ & $0$ & $\Z_2$ & $ 0 $ & $\Z_2$ & $\Z_2$ & $0$ & $\Z_2 $ & $\Z_2^2 $
\end{tabular}
\end{equation}
The ring structure can be written as
\begin{equation} \label{eq:mod2ring}
    H^\bullet(K(\Z, 3),\Z_2) = \Z_2[u,\, \mathrm{Sq}^2 u,\mathrm{Sq}^4 \mathrm{Sq}^2 u,\cdots]\,,
\end{equation}
 where $u \in H^3(K(\Z,3),\Z_2)$ is the mod-2 reduction of $\iota$ and $\mathrm{Sq}^i$ are the Steenrod squares. Note that we have other independent generators $u^2 \in H^6(K(\Z, 3),\Z_2)$, $u \,\mathrm{Sq}^2 u  \in H^8(K(\Z, 3),\Z_2)$ and $u^3  \in H^9(K(\Z, 3),\Z_2)$.
\end{itemize}

In the analysis below, the $\Z_3$-(co)homology of $K(\Z,3)$ will also play a role, although this is not yet apparent. Indeed, for the degrees relevant to our computation, the coefficients $\Omega^{\mathcal S}_q(\mathrm{pt})$ are powers of $\Z$ or $\Z_2$. The reason why $\Z_3$-(co)homology~\cite{Cartan1954EM} nevertheless enters the discussion will become clear shortly.
\begin{itemize}
\item Mod-3 homology
\begin{equation} \label{eq:mod3homology}
    \begin{tabular}{c|ccccccccc}
  $p$ & $0$ & $1$ & $2$ & $3$ & $4$ & $5$ & $6$ & $7$ & $8$  \\ \midrule
$H_p(K(\Z, 3),\Z_3)$ & $\Z_3$ & $0$ & $0$ & $\Z_3$ & $ 0 $ & $0$ & $0$ & $\Z_3$ & $\Z_3$
\end{tabular}
\end{equation}
    \item Mod-3 cohomology
\begin{equation}\label{eq:mod3cohomology}
    \begin{tabular}{c|ccccccccc}
  $p$ & $0$ & $1$ & $2$ & $3$ & $4$ & $5$ & $6$ & $7$ & $8$  \\ \midrule
$H^p(K(\Z, 3),\Z_3)$ & $\Z_3$ & $0$ & $0$ & $\Z_3$ & $ 0 $ & $0$ & $0$ & $\Z_3$ & $\Z_3$
\end{tabular}
\end{equation}
The ring structure can be written as
\begin{equation} \label{eq:mod3ring}
    H^\bullet(K(\Z, 3),\Z_3) = \Z_3[v,\, \mathcal{P}^1 v, \beta_3 \mathcal{P}^1v,\cdots]\,,
\end{equation}
where $v \in H^3(K(\Z,3),\Z_3)$ is the mod-3 reduction of $\iota$, and $\beta_3: H^n(X,\Z_3)\rightarrow H^{n+1}(X,\Z_3)$ is the Bockstein homomorphism associated with the
coefficient sequence $0 \rightarrow \Z_3 \rightarrow \Z_9 \rightarrow \Z_3\rightarrow 0$.

\end{itemize}
For an odd prime $p$, the \emph{Steenrod reduced power operations} $\mathcal{P}^i$ are natural transformations
\begin{equation}
    \mathcal{P}^i : H^q(X, \Z_p) \longrightarrow  H^{q+2 i(p-1)}(X, \Z_p) \,, \qquad i\geqslant 0\,,
\end{equation}
which satisfy properties analogous to those of the Steenrod squares $\mathrm{Sq}^i$. A pedagogical account of these topics may be found in Chapter~VI, Section~15 of~\cite{Bredon}. With such operations one can explore in general the structure of $H^\bullet(K(\Z, n),\Z_p)$ as explained in~\cite{Cartan1954EM}. In our cases, we encounter the particular value of $p=3$ and we will need Theorem~19.7 (Wu) in~\cite{MilnorStasheff1974}
\begin{equation}
 \int_{M_7} \mathcal{P}^1 x = \int_{M_7} x \smile p_1(TM_7) \,,
\end{equation}
for an oriented 7-manifold $M_7$ and a class $x \in H^3(M_7,\Z_3)$.

\subsection{Oriented Bordism Groups} The $E^2$-page can be read off using the integral homology \eqref{eq:integralhomology}, the mod-2 homology \eqref{eq:mod2homology} and the oriented bordism of point \eqref{eq:pointbordism}
\begin{equation} \label{eq:AHSSorientedKZ3}
    E^2_{p,q} = H_p(K(\Z,3),\Omega^{\rm SO}_q({\rm pt})) \Longrightarrow  \Omega^{\rm SO}_{p+q}(K(\Z,3)) \,.
\end{equation}
The relevant results are presented in Table~\ref{tab:E2oriented}, where the columns and rows are counted by $p$ and $q$, respectively. Empty entries in the $E^2$-page are just trivial.

\begin{sseqdata}[ name = E2oriented, classes = { draw = none }, axes type =  frame, scale = 0.8 ]

\class["\Z"](0,0)
\class["\Z"](3,0)
\class["\Z_2"](5,0)
\class["\Z_3"](7,0)
\class["\Z_2"](8,0)
\class["\Z_2"](9,0)

\class["\Z"](0,4)
\class["\Z"](3,4)
\class["\Z_2"](5,4)
\class["\Z_3"](7,4)
\class["\Z_2"](8,4)
\class["\Z_2"](9,4)

\class["\Z_2"](0,5)
\class["\Z_2"](3,5)
\class["\Z_2"](5,5)
\class["\Z_2"](6,5)
\class["\Z_2"](8,5)
\class["\Z_2^2"](9,5)

\class["\Z^2"](0,8)
\class["\Z^2"](3,8)
\class["\Z_2^2"](5,8)
\class["\Z_3^2"](7,8)
\class["\Z_2^2"](8,8)
\class["\Z_2^2"](9,8)

\end{sseqdata}

\begin{table}[htbp]
  \centering
  \printpage[
    name = E2oriented,
    grid = chess,
  ]
  \caption{The $E^2$-page of entries $E^2_{p,q} = H_p(K(\Z,3),\Omega^{\rm SO}_q({\rm pt}))$.}
  \label{tab:E2oriented}
\end{table}
Recall that there is a canonical decomposition
\begin{equation}
    \Omega^{\mathcal{S}}_{i}(K(\Z,3)) \cong \Omega^{\mathcal{S}}_{i}({\rm pt}) \oplus \tilde \Omega^{\mathcal{S}}_{i}(K(\Z,3)) \,,
\end{equation}
which indicates that all differentials involving the $p=0$ column must be trivial. Then, it is easy to see from Table~\ref{tab:E2oriented} that
\begin{align}
  \Omega^{\rm SO}_{i}(K(\Z,3)) \cong \left\{\begin{array}{ll}
\Z & i =0  \\
0 & i =1,2 \\
\Z \cong H_3(K(\Z,3),\Z)& i = 3 \\
\Z \cong \Omega^{\rm SO}_4{(\rm pt)}& i = 4\\
\Z_2^2 \cong \Omega^{\rm SO}_5{(\rm pt)} \oplus H_5(K(\Z,3),\Z) & i = 5\\
0& i = 6
\end{array}\right.\,.
\end{align}

\subsubsection{$\Omega^{\rm SO}_{7}(K(\Z,3))$ and Milnor's Construction}
\label{sec:Milnor}

What is interesting here begins with $i= p+q =7$, and $E^2_{7,0} = H_7(K(\Z,3),\Z) \cong \Z_3$ obviously stabilizes to $E^\infty_{7,0}$ since there are no nontrivial differentials hitting or leaving it, while the possible differential $\td^5$ on the $E^5$-page is also trivial due to the group values
\begin{equation} \label{eq:d5KZ3}
    \td^5 : (E^5_{8,0})_{=H_8(K(\Z,3),\Z) \cong \Z_2} \longrightarrow (E^5_{3,4})_{=H_3(K(\Z,3),\Z) \cong \Z}\,.
\end{equation}
Hence, $E^2_{3,4} = E^5_{3,4} = E^\infty_{3,4} = H_3(K(\Z,3),\Z)$ stabilizes and we get
\begin{equation}
    \Omega^{\rm SO}_{7}(K(\Z,3)) \cong \tilde \Omega^{\rm SO}_{7}(K(\Z,3)) \,,
\end{equation}
sitting in the short exact sequence
\begin{equation} \label{eq:seqOmega7orient}
    0 \longrightarrow   H_3(K(\Z,3),\Z)_{\cong \Z} \xrightarrow{\quad I \quad }   \Omega^{\rm SO}_{7}(K(\Z,3)) \xrightarrow{\quad P \quad } H_7(K(\Z,3),\Z)_{\cong\Z_3} \longrightarrow  0\,.
\end{equation}
To see why we only have the short exact sequence~\eqref{eq:seqOmega7orient}, recall that the AHSS actually computes
\begin{equation}
    \text{Gr}(\Omega^{\rm SO}_{7}(K(\Z,3))) = E^\infty_{0,7} \oplus E^\infty_{1,6} \oplus E^\infty_{2,5} \cdots \oplus E^\infty_{7,0}\,,
\end{equation}
where the same graded group $\text{Gr}(\Omega^{\rm SO}_{7}(K(\Z,3)))$ is by definition also obtained from the filtration of the CW-structure\footnote{For a CW-complex $X$, there exists a sequence of skeletons $\cdots X^{(p)} \subset X^{(p+1)} \cdots$, where $X^{(p)}$ is the $p$-skeleton of $X$. Then for any generalized homology theory $\mathcal{H}_\bullet$, the groups $G_{n,n-p}:= {\rm Im} (\mathcal{H}_n(X^{(p)}) \rightarrow \mathcal{H}_n(X))$ form a filtration
\begin{equation*}
     G_{0,n} \subset G_{1,n-1} \subset G_{2,n-2} \subset \cdots \subset G_{n,0} = \mathcal{H}_n(X)\,.
\end{equation*} The associated graded group is
\begin{equation}
    \text{Gr}(\mathcal{H}_n(X)) :=  G_{0,n} \oplus \frac{G_{1,n-1}}{G_{0,n}} \oplus \frac{G_{2,n-2}}{G_{1,n-1}}\oplus \cdots \oplus \frac{G_{n,0}}{G_{n-1,1}} \,.
\end{equation}} of $K(\Z,3)$
\begin{equation}
    G_{0,7} \subset G_{1,6} \subset G_{2,5} \subset \cdots \subset G_{7,0} = \Omega^{\rm SO}_{7}(K(\Z,3))
\end{equation}
and
\begin{equation}
    \text{Gr}(\Omega^{\rm SO}_{7}(K(\Z,3))) := G_{0,7} \oplus \frac{G_{1,6}}{G_{0,7}} \oplus \frac{G_{2,5}}{G_{1,6}}\oplus \cdots \oplus \frac{G_{7,0}}{G_{6,1}} \,.
\end{equation}
The AHSS actually tells that
\begin{equation}
    E^\infty_{p,q}\;\cong\; G_{p,q}/G_{p-1,q+1}.
\end{equation}
The only two nonvanishing values $E^\infty_{3,4}$ and $E^\infty_{7,0}$ in our result show that the filtration collapses to
\begin{equation}
  G_{0,7}  \cong G_{1,6} \cong G_{2,5} \subset G_{3,4}  \cong G_{4,3} \cong \cdots \cong G_{6,1}\subset G_{7,0}\,.
\end{equation}
The CW-structure of $K(\Z,3)$ tells that $G_{2,5} =\Omega^{\rm SO}_{7}(\rm pt) = 0$. Thus, our AHSS computation gives
\begin{equation}
    \text{Gr}(\Omega^{\rm SO}_{7}(K(\Z,3))) = E^\infty_{3,4} \oplus E^\infty_{7,0} = G_{3,4} \oplus \frac{G_{7,0}}{G_{3,4}} \,.
\end{equation}
However, this alone does not tell if the short exact sequence~\eqref{eq:seqOmega7orient} splits; we need to solve this extension problem with more geometric considerations as follows.

Surjectivity of the \emph{edge} homomorphism $P$ indicates that any class in $H_7(K(\Z,3),\Z)$ can be realized as the image $f(M_7)$ of a map $f : M_7 \longrightarrow K(\Z,3)$, where $M_7$ is an oriented seven manifold and $[M_7, f]$ represents a class in $\Omega^{\rm SO}_{7}(K(\Z,3))$.
Furthermore, recall~\eqref{eq:mod3cohomology}, this group is purely 3-torsion, and from the knowledge of mod-3 cohomology ring of $K(\Z,3)$, we see that it can be detected by the generator $\mathcal{P}^1v$ of $ H^7(K(\Z,3),\Z_3) \cong \Z_3$. 
In other words, we have the usual Kronecker pairing
\begin{equation}
\begin{aligned} \label{eq:Z3bordismpairing}
    H_7(K(\Z,3),\Z_3) \times H^7(K(\Z,3),\Z_3) &\longrightarrow \Z_3 \\
  [M_7, f] \times \mathcal{P}^1v &\longmapsto \int_{M_7} f^*(\mathcal{P}^1v) = \int_{M_7} f^*(v) \smile p_1(TM_7)  \,.
\end{aligned}
\end{equation}
With this pairing at hand, we can detect if an element $[M_7, f]  \in \Omega^{\rm SO}_{7}(K(\Z,3))$ is sent to the generator of $H_7(K(\Z,3),\Z)$, which happens if and only if the pairing evaluates to $\pm 1 \mod 3$.

Now for the extension problem, if we can
\begin{itemize}
    \item construct a class $[\tilde M_7, \tilde f] \in \Omega^{\rm SO}_{7}(K(\Z,3))$ mapped to the generator of $H_7(K(\Z,3),\Z)$;
    \item and show that $[\tilde M_7, \tilde f]$ is an element of infinite order in $\Omega^{\rm SO}_{7}(K(\Z,3))$,
\end{itemize}
then the sequence~\eqref{eq:seqOmega7orient} is non-split. As for the latter point, we will need a $\Z$-valued bordism invariant as an extra input. Such a bordism invariant is easy to guess and it is
\begin{equation}
\begin{aligned} \label{eq:Zbordisminvariant}
     \Phi :  \Omega^{\rm SO}_{7}(K(\Z,3)) &\longrightarrow \Z \\
    [M_7,f] & \longmapsto \int_{M_7}  f^*(\iota) \smile p_1(TM_7) \,.
\end{aligned}
\end{equation}

\paragraph{Milnor's construction.} The desired class $[\tilde M_7, \tilde f]$ can be obtained by a construction due to Milnor in his famous work~\cite{MilnorSphere} on the exotic spheres. Consider the total space of a $S^3$-bundle over $S^4$ with structure group $SO(4)$
\begin{equation}
\begin{tikzcd}
S^3 \arrow[r, hookrightarrow] & M_{n_L,n_R} \arrow[d] \\
                              & S^4
\end{tikzcd}
\end{equation}
where, in general, isomorphism classes of such sphere bundles are classified by
\begin{equation}
    \pi_3(SO(4)) \cong \Z \oplus \Z\,,
\end{equation}
and hence by two integers $n_L$ and $n_R$. It is conventional to write the total space as $M_{n_L,n_R}$ and in physics terms the two integers are just instanton numbers of the 2-fold covering group of $SO(4)$
\begin{equation}
    \frac{SU(2)_L \times SU(2)_R}{ \Z_2} \cong SO(4).
\end{equation}
For such sphere bundles, their characteristic classes are well-understood (see, e.g.,~\cite{mcenroe2015milnor, CROWLEY2003363}). Let $\alpha \in H^4(S^4,\Z) \cong \Z$ be the generator, then the Euler class $e$ and the first Pontryagin class\footnote{Unless otherwise specified, $p_i$ denotes the Pontryagin classes of the tangent bundles of the manifolds under consideration.} $p_1$ can be expressed as
\begin{equation}
    e = (n_L - n_R) \alpha\,, \qquad
    p_1 = 2 (n_L + n_R) \alpha \,.
\end{equation}
Note that Milnor showed that $M_{n_L,n_R}$ is a \emph{topological sphere} (homeomorphic to $S^7$) if and only if $n_L - n_R = \pm 1$.

For the present purpose, take $\tilde M_7 = M_{1,1}$, for which $e =0$ and $p_1 =4 \alpha$.\footnote{We thank Yuji Tachikawa for suggesting this geometric construction as an approach to the extension problem.} Its cohomology follows from the Gysin sequence:
\begin{align}
  H^i(\tilde M_7, \Z)=\left\{\begin{array}{ll}
\Z & \text{for} \; i =0,3,4,7  \\
0 & \text{otherwise} \\
\end{array}\right.\,.
\end{align}
The map $\tilde f: \tilde M_7 \longrightarrow K(\Z,3)$ is such that $\tilde f^*(\iota)$ equals the generator of $H^3(\tilde M_7, \Z)$ which comes from the generator of $ H^3(S^3,\Z) \cong \Z$ on the fiber. $\tilde M_7$ is obviously oriented, so we have a class $[\tilde M_7, \tilde f] \in \Omega^{\rm SO}_{7}(K(\Z,3))$.
Since
\begin{equation}
      [\tilde M_7, \tilde f] \times \mathcal{P}^1 v \longmapsto \int_{\tilde M_7} \tilde f^*(\mathcal{P}^1 v) = \int_{\tilde M_7 } \tilde f^*(v) \smile p_1(T\tilde M_7)  = 1 \mod 3\,,
\end{equation}
we see that $[\tilde M_7, \tilde f]$ is mapped to the generator of $H_7(K(\Z,3),\Z)$.
Now using~\eqref{eq:Zbordisminvariant}
\begin{equation}
    \Phi ( [\tilde M_7, \tilde f]) =  \int_{\tilde M_7}  \tilde f^*(\iota) \smile p_1(T\tilde M_7) = 4 \,,
\end{equation}
we can conclude that $[\tilde M_7, \tilde f]$ is of infinite order in $\Omega^{\rm SO}_{7}$ and the sequence~\eqref{eq:seqOmega7orient} does not split and this proves the following result
\begin{theorem} \label{theoremOmega7}
 $ \boxed{\Omega^{\rm SO}_{7}(K(\Z,3))  \cong \Z}$\,.
\end{theorem}

On the other hand, we would also like to understand the injection $I$ from $H_3(K(\Z,3),\Z)$ to $\Omega^{\rm SO}_{7}(K(\Z,3))$ in~\eqref{eq:seqOmega7orient}. By the Hurewicz theorem, $H_3(K(\Z,3),\Z) \cong \pi_3(K(\Z,3)) =[S^3, K(\Z,3)]$. Let $\varphi : S^3  \rightarrow K(\Z,3)$ represent the generator of $H_3(K(\Z,3),\Z) \cong \Z$. Then, the geometric realization of the map $I$ is
\begin{equation}
    \begin{aligned} \label{eq:itheinjection}
        I :  H_3(K(\Z,3),\Z) &\longrightarrow   \Omega^{\rm SO}_{7}(K(\Z,3)) \\
        \varphi & \longmapsto [S^3 \times \mathbb{CP}^2, \varphi \circ {\rm proj}_{S^3} ] \,,
    \end{aligned}
\end{equation}
where we used the fact that $E^\infty_{3,4} = H_3(K(\Z,3),\Omega^{\rm SO}_4({\rm pt}))= H_3(K(\Z,3),\Z)$ and that $[\mathbb{CP}^2]$ is the standard geometric generator of $\Omega^{\rm SO}_4({\rm pt})$, and ${\rm proj}_{S^3} :  S^3 \times \mathbb{CP}^2 \rightarrow S^3$ is the projection to $S^3$.

As a crosscheck of our proposal, we observe that $P(3 \times [\tilde M_7, \tilde f]) = 0 \in H_7(K(\Z,3),\Z)$, and hence $3 \times [\tilde M_7, \tilde f] \in \ker(P)=\im(I)$. Evaluating the image in~\eqref{eq:itheinjection}, we have
\begin{equation}
    \Phi ([ S^3 \times \mathbb{CP}^2, \varphi \circ {\rm proj}_{S^3} ]) = \int_{S^3 \times \mathbb{CP}^2}  (\varphi \circ {\rm proj}_{S^3})^*(\iota) \smile p_1 = \int_{S^3} \varphi^*(\iota)  \int_{\mathbb{CP}^2} p_1 = 3\,.
\end{equation}
Since $\Phi([\tilde M_7,\tilde f])=4$, we conclude that the element $4 \times \varphi \in H_3(K(\Z,3),\Z)$ is sent to $3 \times [\tilde M_7, \tilde f]$ by $I$.
Extensions of $\Z_3$ by $\Z$ are classified by the Abelian group $\mathrm{Ext}^1_{\Z}(\Z_3,\Z) \cong \Z_3$. 
The integer $1 = 4 \mod 3$ indicates the nontrivial extension class $1\in \mathrm{Ext}^1_{\Z}(\Z_3,\Z)$. This is consistent with the non-splitness of the short exact sequence~\eqref{eq:seqOmega7orient}.

\subsubsection{$\Omega^{\rm SO}_{8}(K(\Z,3))$ and the Extension Problem}
For $p+q = 8$, the entry $E^2_{8,0} = H_8(K(\Z,3),\Z) \cong \Z_2$ survives to the final page, for there is no nontrivial differential from it to the 7-th diagonal.

Next, we have the following lemma:
\begin{lemma} \label{lemma1}
The two possible differentials
\begin{equation}
     \td^2 : (E^2_{5,4})_{= H_5(K(\Z,3),\Z) \cong \Z_2} \longrightarrow (E^2_{3,5})_{ = H_3(K(\Z,3),\Z_2) \cong \Z_2}
\end{equation}
and
\begin{equation}
     \td^6 : (E^6_{9,0})_{=H_9(K(\Z,3),\Z) \cong \Z_2} \longrightarrow (E^6_{3,5})_{=H_3(K(\Z,3),\Z_2)\cong \Z_2}
\end{equation} in the Atiyah--Hirzebruch spectral sequence \eqref{eq:AHSSorientedKZ3} for $\Omega^{\rm SO}_{\bullet}(K(\Z,3))$ both vanish.
\end{lemma}

Abstractly speaking, in the AHSS for oriented bordism, the differentials are shown~\cite{Maunder1964} to be induced by the stable cohomology operations determined by the Postnikov $k$-invariant of the Thom spectrum $MSO$~\cite{Thom1954}. Here we are only dealing with purely 2-torsion groups, the only possible nonzero differential is controlled by 2-primary $k$-invariants. Since these 2-primary $k$-invariants are known to be trivial~\cite{Taylor1976,Troue1966}, the corresponding $\td^2$ and $\td^6$ vanish.\footnote{%
If one is interested only in $(I_{\Z}\Omega^{\rm SO})^7(K(\Z,3))$, the cohomological AHSS for the Anderson dual gives a shorter computation. Using $(I_{\Z}\Omega^{\rm SO})^4({\rm pt})=\Z\langle p_1/3\rangle$, the only nonzero differential relevant for total degree seven is
\begin{equation*}
\begin{aligned}
\td_5^{3,4}: E_5^{3,4}
&\cong\Z\langle \iota\,\frac{p_1}{3}\rangle
\longrightarrow E_5^{8,0}
\cong\Z_3\langle\kappa\rangle,\\
\td_5^{3,4}\left(\iota\,\frac{p_1}{3}\right)
&=\beta\mathcal{P}^1\rho(\iota)=\kappa\,.
\end{aligned}
\end{equation*}
Here $\rho$ denotes reduction modulo $3$, and $\beta$ is the integral Bockstein associated with $0\to\Z\xrightarrow{\times 3}\Z\to\Z_3\to 0$. The underlying $k$-invariant $\beta\mathcal{P}^1\rho$ was originally determined by Thom~\cite{Thom1954}; the corresponding cohomological AHSS differential $\td_5^{3,4}$ in modern notation is discussed and computed in~\cite{SPTandSteenrod}. Since this differential is surjective, $E_\infty^{3,4}=\ker(\td_5^{3,4})$ is generated by $\iota p_1$, which gives $(I_{\Z}\Omega^{\rm SO})^7(K(\Z,3))\cong\Z \langle\iota p_1\rangle$. We nevertheless use the homological AHSS in the main text because we also want to determine the bordism group itself, resolve its extension problem, and identify its geometric generators.\label{footnote:cohomologyAHSS}}

For a more concrete proof of Lemma~\ref{lemma1}, we just use the known result~\cite{Wan:2018bns} $\tilde \Omega^{\rm SO}_{7}(K(\Z,2))\cong \Z_2$. The key idea, following~\cite{Joyce:2025jyd}, is to consider the maps between spectral sequences induced by the suspension isomorphism\footnote{At times we work with reduced (co)homology $\tilde H$, since it satisfies the Eilenberg--Steenrod axioms, including the suspension property. Because reduced and unreduced (co)homology differ only in degree 0, we will use $H$ and $\tilde H$ interchangeably in positive degrees. The same convention applies to reduced bordism.}
 $\tau : \tilde H^\bullet(K(\Z,2),\Z) \rightarrow  \tilde H^{\bullet+1}(\Sigma K(\Z,2),\Z)$.
\begin{proof}[Proof of Lemma~\ref{lemma1}]
Let $c_1$ (first Chern class) denote the generator of $\tilde H^2(K(\Z,2),\Z)$. Then $\tau(c_1) \in \tilde H^{3}(\Sigma K(\Z,2),\Z) \cong \Z$ is the generator. On the other hand, we have the property of Eilenberg--Mac Lane spaces
\begin{equation}
    \begin{aligned}
         \tilde H^{3}(\Sigma K(\Z,2),\Z) &\xrightarrow{\;\cong\;} [\Sigma K(\Z,2), K(\Z,3)] \\
         \tau(c_1) &\longmapsto \chi \,,
    \end{aligned}
\end{equation}
which determines a map $\chi : \Sigma K(\Z,2)\rightarrow K(\Z,3) $ up to homotopy and with the property that
\begin{equation}
    \begin{aligned}
         \chi^*:\tilde H^{3}(K(\Z,3),\Z) &\xrightarrow{\;\cong\;} \tilde H^{3}(\Sigma K(\Z,2),\Z) \\
          \iota &\longmapsto \tau(c_1) \,.
    \end{aligned}
\end{equation}
We also have dually induced maps $\chi_*$ on the (generalized) homology side.
Since the reduced bordism is also a generalized homology theory, we have a suspension isomorphism $\tilde \Omega^{\rm SO}_{i}(K(\Z,2))\xrightarrow{\;\cong\;}\tilde \Omega^{\rm SO}_{i+1}(\Sigma K(\Z,2))$. It can be composed with $\chi_*$ to get a morphism between spectral sequences
\begin{equation}
    \tilde \Omega^{\rm SO}_{i+1}( K(\Z,2)) \xrightarrow{ \;\cong\;} \tilde \Omega^{\rm SO}_{i+1}(\Sigma K(\Z,2)) \xrightarrow{ \;\chi_*\;} \tilde \Omega^{\rm SO}_{i+1}( K(\Z,3)) \,,
\end{equation}
which we will still refer to as $\chi_*$.
Recall the AHSS for $\tilde \Omega^{\rm SO}_{\bullet}( K(\Z,2))$ (here the entries will be denoted as $\bar E^r_{p,q}$)
 \begin{equation} \label{AHSSorientedKZ2}
    \bar E^2_{p,q} = \tilde H_p(K(\Z,2), \Omega^{\rm SO}_q({\rm pt})) \Longrightarrow  \tilde \Omega^{\rm SO}_{p+q}(K(\Z,2)) \,,
\end{equation}
whose $E^2$-page is presented in Table~\ref{tab:E2orientedKZ2}.
\begin{sseqdata}[ name = E2orientedKZ2, classes = { draw = none }, axes type =  frame, scale = 0.8 ]

\class(0,0)

\class["\Z"](2,0)
\class["\Z"](4,0)
\class["\Z"](6,0)
\class["\Z"](8,0)

\class["\Z"](2,4)
\class["\Z"](4,4)
\class["\Z"](6,4)
\class["\Z"](8,4)

\class["\Z_2"](2,5)
\class["\Z_2"](4,5)
\class["\Z_2"](6,5)
\class["\Z_2"](8,5)

\class(8,7)

\end{sseqdata}

\begin{table}[htbp]
  \centering
  \printpage[
    name = E2orientedKZ2,
    grid = chess,
  ]
  \caption{The $E^2$-page of the reduced oriented bordism $\tilde \Omega^{\rm SO}_\bullet(K(\Z,2))$.}
  \label{tab:E2orientedKZ2}
\end{table}

The morphism $\chi_*$ appears on the $E^r$-pages as
\begin{equation}
    \chi_*: \bar E^r_{p,q}(K(\Z,2)) \longrightarrow  E^r_{p+1,q}(K(\Z,3))\,,
\end{equation}
and hence in particular we have the following commutative diagram
\begin{equation}
\begin{tikzcd} \label{eq:comdiag1}
\bar E^2_{4,4}(K(\Z,2)) \cong \Z  \langle  \hat{c}^2_1 \alpha \rangle \arrow[r, "\chi_*"] \arrow[d,"\bar \td^2"'] & E^2_{5,4}(K(\Z,3)) \cong \Z_2 \langle \hat \iota^2  \alpha \rangle \arrow[d,"\td^2"] \\
\bar E^2_{2,5}(K(\Z,2)) \cong  \Z_2 \langle  \hat{\bar c}_1 \omega \rangle \arrow[r,"\chi_*"'] & E^2_{3,5}(K(\Z,3)) \cong \Z_2 \langle \hat u  \omega  \rangle \,,
\end{tikzcd}
\end{equation}
where $\bar c_1$ is the mod-2 reduction of $c_1$. The generators of the various homology groups are explicitly given by using the ``hat'' to indicate the homology dual of various cohomology classes, respectively. Following Notation 15.1 in~\cite{Joyce:2025jyd}, homology products mean the homology class dual of cohomology cup products with the chosen basis, and $\alpha$ is simply $\mathbb{CP}^2$ while $\omega = SU(3)/SO(3)$ is the Wu manifold. Clearly, $\hat c_1 \mapsto \hat \iota$, $\hat{\bar c}_1 \mapsto \hat u$ and $\hat c_1^2 \mapsto \hat \iota^2$. Hence, the upper horizontal map $\chi_*$ in \eqref{eq:comdiag1} is surjective, while the lower horizontal map $\chi_*$ is an isomorphism. We know from the result $\tilde \Omega_7^{\rm SO}(K(\Z,2)) = \Z_2$ that $\bar \td^2$ vanishes, due to commutativity of the above diagram, $\td^2$ on the right must vanish as well.

For $\td^6$, the same argument as above would not work (see the remark below). Fortunately, we can proceed geometrically. By definition, $K(\Z,3)$'s CW-construction starts from the 3-cell and hence we can immediately tell
\begin{equation}
\begin{aligned}
        G_{2,6} &= {\rm Im} \left(\tilde \Omega_8^{\rm SO}(K(\Z,3)^{(2)}) \rightarrow \tilde \Omega_8^{\rm SO}((K(\Z,3))\right) \\
        &= {\rm Im} \left(\tilde \Omega_8^{\rm SO}({\rm pt}) \rightarrow \tilde \Omega_8^{\rm SO}((K(\Z,3))\right) = 0
\end{aligned}
\end{equation}
On the final page we have
\begin{equation}
    E^\infty_{3,5} \cong G_{3,5}/G_{2,6} =  G_{3,5}\,.
\end{equation}
Whether $\td^6$ is an isomorphism between $\Z_2$'s or vanishes is equivalent to whether $G_{3,5}$ is zero. It is enough to find a nonzero class in $\tilde \Omega^{\rm SO}_8(K(\Z,3))$ which lies in the image of $\tilde \Omega^{\rm SO}_8(K(\Z,3)^{(3)}) = \tilde \Omega^{\rm SO}_8(S^3)$.
Following the construction in \eqref{eq:itheinjection}, we choose the class $[S^3 \times SU(3)/SO(3), f=\varphi \circ {\rm proj}_{S^3}]$ and evaluate the bordism invariant
\begin{equation}
\begin{aligned} \label{eq:oribordisminva1}
           \Phi_1 ([ S^3 \times SU(3)/SO(3), \varphi \circ {\rm proj}_{S^3} ]) :&= \int_{S^3 \times SU(3)/SO(3)}  (\varphi \circ {\rm proj}_{S^3})^*(u) \smile w_2\smile w_3 \\
           & = \int_{S^3} \varphi^*(u)  \int_{SU(3)/SO(3)} w_2\smile w_3 = 1 \mod 2\,.
\end{aligned}
\end{equation}
$\Phi_1$ here represents the mod-2 characteristic class $u \, w_2 \,w_3$, which is obviously a bordism invariant for $\tilde \Omega^{\rm SO}_8(K(\Z,3))$. We can therefore conclude that $\td^6$ is also zero.
\end{proof}

\paragraph{Remark.} One may also want to play the same suspension trick on the $E^6$-page to determine if $\td^6$ is zero. However, this does not work due to the following cohomology computations.
Consider the commutative diagram
\begin{equation}
\begin{tikzcd} \label{eq:comdiag2}
\bar E^6_{8,0}(K(\Z,2)) \cong \Z  \langle  \hat{c}^4_1  \rangle \arrow[r, "\chi_*"] \arrow[d,"\bar \td^6"'] & E^6_{9,0}(K(\Z,3)) \cong \Z_2 \langle \hat \iota^3   \rangle \arrow[d,"\td^6"] \\
\bar E^6_{2,5}(K(\Z,2)) \cong  \Z_2 \langle  \hat{\bar c}_1 \omega \rangle \arrow[r,"\chi_*"'] & E^6_{3,5}(K(\Z,3)) \cong \Z_2 \langle \hat u  \omega  \rangle \,.
\end{tikzcd}
\end{equation}
Now for the mod-2 cohomologies, we have $\bar c^4_1 = \mathrm{Sq}^4 \mathrm{Sq}^2 \bar c_1$, but $u^3$ (recall that it is the mod-2 reduction of $\iota^3$) and $\mathrm{Sq}^4 \mathrm{Sq}^2 u$ are independent generators, they generate $H^9(K(\Z,3),\Z_2) \cong \Z_2^2$. As the induced morphisms preserve Steenrod squares, we have $\bar c^4_1 \mapsto  \mathrm{Sq}^4 \mathrm{Sq}^2 u$ instead of $u^3$. Hence, the upper $\chi_*$ in \eqref{eq:comdiag2} is zero, and we can not conclude from the commutativity of the diagram and the vanishing of $\bar \td^6$ that $\td^6=0$.

\paragraph{The extension problem for $\tilde \Omega^{\rm SO}_{8}(K(\Z,3))$.} It follows from Lemma~\eqref{lemma1} that
\begin{equation}
    \text{Gr}( \tilde \Omega^{\rm SO}_{8}(K(\Z,3))) = E^\infty_{3,5} \oplus E^\infty_{8,0} = G_{3,5} \oplus \frac{G_{8,0}}{G_{3,5}} \,,
\end{equation}
yielding another extension problem of the following short exact sequence
\begin{equation} \label{eq:seqOmega8orient}
    0 \longrightarrow   H_3(K(\Z,3),\Z_2)_{\cong \Z_2} \xrightarrow{\quad j \quad }   \tilde \Omega^{\rm SO}_{8}(K(\Z,3)) \xrightarrow{\quad q \quad } H_8(K(\Z,3),\Z)_{\cong\Z_2} \longrightarrow  0\,.
\end{equation}
Hence, $\tilde \Omega^{\rm SO}_{8}(K(\Z,3))$ is either $\Z_2 \oplus \Z_2$ or $\Z_4$ depending on whether \eqref{eq:seqOmega8orient} splits.
This can be solved geometrically by looking at the geometric representative  $[SU(3), \phi]$, where $\phi^*$ pulls back $\iota \in H^3(K(\Z,3),\Z)$ to the generator of $H^3(SU(3),\Z) \cong \Z$. Recall the generator $\Phi_2:=u\,\mathrm{Sq}^2 u$ of $H^8(K(\Z,3),\Z_2)_{\cong\Z_2}$, which is another bordism invariant.
Then (see Theorem 3.5 in~\cite{Joyce:2025jyd})
\begin{equation}
    \Phi_2([SU(3), \phi]) = \int_{SU(3)} \phi^*(u\,\mathrm{Sq}^2 u) = 1 \mod 2\,,
\end{equation}
while $\Phi_2([ S^3 \times SU(3)/SO(3), \varphi \circ {\rm proj}_{S^3} ]) = \int_{S^3} \varphi^*(u\,\mathrm{Sq}^2 u) \times 1= 0 \mod 2$.
At the same time,
\begin{equation}
    \Phi_1([SU(3), \phi]) = \int_{SU(3)} \phi^*( u ) \smile  w_2\smile w_3= 0 \mod 2\,,
\end{equation}
because $SU(3)$ is a spin manifold, on which $w_2$ vanishes (in fact, it is the geometric generator~\cite{Joyce:2025jyd} of $\tilde \Omega^{\rm Spin}_8(K(\Z,3)) \cong \Z_2$).
Recall from~\eqref{eq:oribordisminva1} that $\Phi_1([ S^3 \times SU(3)/SO(3), \varphi \circ {\rm proj}_{S^3} ])= 1 \mod 2$, and we can combine the two independent $\Z_2$-valued bordism invariants to the surjective map $(\Phi_1,\Phi_2): \tilde \Omega^{\rm SO}_{8}(K(\Z,3))\longrightarrow \Z_2 \oplus \Z_2$, which turns out to be an isomorphism. The short exact sequence \eqref{eq:seqOmega8orient} splits and this proves
\begin{theorem}
  $ \boxed{ \tilde \Omega^{\rm SO}_{8}(K(\Z,3))  \cong \Z_2 \oplus \Z_2}$\,.
  \label{theorem:3.3}
\end{theorem}

\subsection{Spin Bordism Groups}
The results for spin bordism can also be found in the recent mathematical literature~\cite{Joyce:2025jyd}, where the reduced spin bordism groups for various spaces, including $K(\Z,3)$, are computed solely using the AHSS. We nevertheless review them here.
The $E^2$-page for the spin bordism group computation is presented in Table~\ref{tab:E2spin}, obtained from the known homology groups~\eqref{eq:integralhomology} and \eqref{eq:mod2homology} together with the spin bordism groups of a point~\eqref{eq:pointbordism}.
For $i\leqslant 3$, we can easily see that
\begin{align}
  \Omega^{\rm Spin}_{i}(K(\Z,3))=\left\{\begin{array}{ll}
\Z & i =0  \\
\Z_2 \cong\Omega^{\rm Spin}_{i}(\rm pt) & i =1,2 \\
\Z \cong H_3(K(\Z,3),\Z)& i = 3 \\
\end{array}\right.\,.
\end{align}

\begin{sseqdata}[ name = E2spin, classes = { draw = none }, axes type =  frame, scale = 0.8 ]

\class["\Z"](0,0)
\class["\Z"](3,0)
\class["\Z_2"](5,0)
\class["\Z_3"](7,0)
\class["\Z_2"](8,0)
\class["\Z_2"](9,0)

\class["\Z_2"](0,1)
\class["\Z_2"](3,1)
\class["\Z_2"](5,1)
\class["\Z_2"](6,1)
\class["\Z_2"](8,1)
\class["\Z_2^2"](9,1)

\class["\Z_2"](0,2)
\class["\Z_2"](3,2)
\class["\Z_2"](5,2)
\class["\Z_2"](6,2)
\class["\Z_2"](8,2)
\class["\Z_2^2"](9,2)

\class["\Z"](0,4)
\class["\Z"](3,4)
\class["\Z_2"](5,4)
\class["\Z_3"](7,4)
\class["\Z_2"](8,4)
\class["\Z_2"](9,4)

\class["\Z^2"](0,8)
\class["\Z^2"](3,8)
\class["\Z_2^2"](5,8)
\class["\Z_3^2"](7,8)
\class["\Z_2^2"](8,8)
\class["\Z_2^2"](9,8)
\end{sseqdata}

\begin{table}[htbp]
  \centering
  \printpage[
    name = E2spin,
    grid = chess,
  ]
  \caption{The $E^2$-page of entries $E^2_{p,q} = H_p(K(\Z,3),\Omega^{\rm Spin}_q({\rm pt}))$.}
  \label{tab:E2spin}
\end{table}
For $i = p+ q = 4$, the differential
\begin{equation}  \label{eq:d2iso}
     \td^2 : (E^2_{5,0})_{=H_5(K(\Z,3),\Z) \cong \Z_2} \longrightarrow (E^2_{3,1})_{=H_3(K(\Z,3),\Z_2) \cong \Z_2}
\end{equation}
is an isomorphism. In general, the differential $\td^2 : H_p(X,  \Omega^{\rm Spin}_1({\rm pt})) \rightarrow H_{p-2}(X, \Omega^{\rm Spin}_2({\rm pt}))$ is shown~\cite{Teichner1993Signature} to be the dual $(\mathrm{Sq}^2)^*$ of $\mathrm{Sq}^2: H^{p-2}(X,  \Z_2) \rightarrow H^p(X, \Z_2)$, while the differential $\td^2 : H_p(X,  \Omega^{\rm Spin}_0({\rm pt})) \rightarrow H_{p-2}(X, \Omega^{\rm Spin}_1({\rm pt}))$ turns out to be the mod $2$ reduction composed with $(\mathrm{Sq}^2)^*$. For other known differentials, the reader may refer to~\cite{ABPspinbordismring}. In our case, we have [cf.~Eq.~\eqref{eq:mod2ring}]
\begin{equation}
    H_5(K(\Z,3),\Z) \xrightarrow{\quad \text{mod} \;2 \quad } H_5(K(\Z,3),\Z_2) \xrightarrow{\quad (\mathrm{Sq}^2)^* \quad } H_3(K(\Z,3),\Z_2) \,,
\end{equation}
\begin{equation}
    \begin{aligned}
        \mathrm{Sq}^2 :  H_3(K(\Z,3),\Z_2) &\xrightarrow{\quad \cong \quad }   H_5(K(\Z,3),\Z_2)\\
        u & \longmapsto \mathrm{Sq}^2 u \,.
    \end{aligned}
\end{equation}
Thus, $E^2_{3,1}$ is eliminated at the $E^2$-page and
\begin{equation}
     \Omega^{\rm Spin}_4(K(\Z,3)) \cong \Z  \cong \Omega^{\rm Spin}_4{(\rm pt)} \,.
\end{equation}

The isomorphism~\eqref{eq:d2iso} also indicates that $E^2_{5,0}$ is eliminated. From the first row to the second row of the $E^2$-page we have the isomorphism
\begin{equation}
     \td^2 = (\mathrm{Sq}^2)^*: (E^2_{5,1})_{=H_5(K(\Z,3),\Z_2) \cong \Z_2} \xrightarrow{\quad \cong \quad }   (E^2_{3,2})_{=H_3(K(\Z,3),\Z_2) \cong \Z_2}\,,
\end{equation}
and it follows that both $E^2_{5,1}$ and $E^2_{3,2}$ are eliminated. As a consequence, we have
\begin{equation}
     \Omega^{\rm Spin}_5(K(\Z,3)) =0 \qquad \text{and} \qquad \Omega^{\rm Spin}_6(K(\Z,3))=0 \,.
\end{equation}

\subsubsection{$\Omega^{\rm Spin}_7(K(\Z,3))$ and $\Omega^{\rm Spin}_8(K(\Z,3))$}
The entry $E^2_{7,0} = H_7(K(\Z,3),\Z) = E^\infty_{7,0}$ stabilizes to the final page as there is no nontrivial group homomorphism from $E^2_{7,0}=\Z_3$ to $E^2_{5,1}=\Z_2$.
The differential $\td^2: E^2_{8,0} \rightarrow E^2_{6,1}$ must be of the form
\begin{equation} \label{eq:triviald280}
    H_8(K(\Z,3),\Z) \xrightarrow{\quad \text{mod} \;2 \quad } H_8(K(\Z,3),\Z_2) \xrightarrow{\quad (\mathrm{Sq}^2)^* \quad } H_6(K(\Z,3),\Z_2) \,.
\end{equation}
The relevant Steenrod square has following action on the generator $u^2 \in H^6(K(\Z,3),\Z_2)$
\begin{equation}
\begin{aligned}
       \mathrm{Sq}^2 u^2 &= \mathrm{Sq}^0 u \smile \mathrm{Sq}^2 u + \mathrm{Sq}^1 u \smile \mathrm{Sq}^1 u + \mathrm{Sq}^2 u \smile \mathrm{Sq}^0 u \\
       &= 2 \times (u \smile \mathrm{Sq}^2 u) + 0 = 0 \mod 2\,.
\end{aligned}
\end{equation}
The vanishing of this differential leads to $E^\infty_{6,1} = E^2_{6,1}$. We can move on to the next nonvanishing entry on the $p+q=7$ diagonal, that is $E^2_{5,2}$, it survives obviously to the $E^3$-page.

By Theorem 3.5 of~\cite{Joyce:2025jyd}, the following differential from $E^3_{8,0}$ to $E^3_{5,2}$ vanishes:
\begin{equation} \label{eq:triviald380}
     \td^3 : (E^3_{8,0})_{=H_8(K(\Z,3),\Z) \cong \Z_2} \longrightarrow (E^3_{5,2})_{=H_5(K(\Z,3),\Z_2) \cong \Z_2}\,.
\end{equation}
On the other hand, the other differential relevant for $\Omega^{\rm Spin}_8(K(\Z,3))$ is an isomorphism:
\begin{equation} \label{eq:isod390}
       \td^3 : (E^3_{9,0})_{=H_9(K(\Z,3),\Z) \cong\Z_2} \longrightarrow (E^3_{6,2})_{=H_6(K(\Z,3),\Z_2) \cong\Z_2}\,.
\end{equation}
It follows immediately that $\tilde \Omega^{\rm Spin}_8(K(\Z,3)) = E^\infty_{8,0} \cong \Z_2$.

As for $\tilde \Omega^{\rm Spin}_7(K(\Z,3))$, we note that in the proof of Theorem 3.5 the authors used the suspension trick relating the spin bordism groups of $K(\Z,3)$ to the results~\cite{Stong} of $K(\Z,4)$
\begin{equation} \label{eq:spinbordismokfKZ4}
    \begin{tabular}{c|cccccccccc}
  $n$ & $0$ & $1$ & $2$ & $3$ & $4$ & $5$ & $6$ & $7$ & $8$  & $9$ \\ \midrule
$\tilde \Omega^{\rm Spin}_n(K(\Z,4))$ & $0$ & $0$ & $0$ & $0$ & $\Z$ & $0$ & $0$ & $0$ & $ \Z^2$ & $\Z_2$
\end{tabular}
\end{equation}
In this case, the suspension trick offers even more than just the determination of some differentials, it remarkably tells us how the filtration is formed from $E^\infty_{7,0}=\Z_3$, $E^\infty_{6,1}=\Z_2$,  $E^\infty_{5,2}=\Z_2$ and $E^\infty_{3,4}=\Z$. That is,
\begin{equation} \label{eq:filtrationSpinKZ}
0\subset G_{3,4} \cong \Z \subset G_{5,2} \cong \Z \subset G_{6,1}\cong \Z\subset G_{7,0}
=\tilde \Omega^{\rm Spin}_7(K(\Z,3)) \cong \Z\,,
\end{equation}
with the graded structure
\begin{equation}
\begin{aligned}
\text{Gr}\big(\tilde \Omega^{\rm Spin}_7(K(\Z,3))\big)
&= E^\infty_{3,4}\oplus E^\infty_{5,2}\oplus E^\infty_{6,1}\oplus E^\infty_{7,0} \\
&=G_{3,4}\oplus \frac{G_{5,2}}{G_{3,4}}\oplus \frac{G_{6,1}}{G_{5,2}}\oplus \frac{G_{7,0}}{G_{6,1}} \\
&\cong \Z\oplus \Z_2\oplus \Z_2\oplus \Z_3.
\end{aligned}
\end{equation}

\section{Physical Examples of Anomalies for $U(1)$ 1-form Symmetries}
\label{sec:physics}

In this section, we discuss physical examples with different types of anomalies for $U(1)$ 1-form symmetries in bosonic and fermionic theories in various dimensions. We do not include time-reversal symmetry, for which one would instead use unoriented or pin bordism. For convenience, we recall the result of bordism groups from Table~\ref{tab:bordism-KZ3-summary}:
\begin{equation}
    \begin{tabular}{c|ccccccccc} \label{eq:tab1informula}
  $n$ & $0$ & $1$ & $2$ & $3$ & $4$ & $5$ & $6$ & $7$ & $8$  \\ \midrule
$\Omega^{\rm SO}_n(K(\Z,3))$ & $\Z$ & $0$ & $0$ & $\Z$ & $\Z$ & $\Z_2^2$ & $0$ & $ \Z$ & $  \Z_2^2  \oplus \Z^2 $   \\  \midrule
$\Omega^{\rm Spin}_n(K(\Z,3))$ & $\Z$ & $\Z_2$ & $\Z_2$ & $\Z$ & $\Z$ & $0$ & $0$ & $\Z$ & $ \Z_2  \oplus \Z^2$
\end{tabular}
\end{equation}
Using the universal coefficient theorem~\eqref{eq:UCTforIOMEGA}, we can read off the values of the invertible phases $(I_{\Z} \Omega^{\mathcal{S}})^{n}(K(\Z,3))$ as follows:\footnote{Although we do not determine the groups $\Omega^{\mathrm{SO}/\mathrm{Spin}}_{9}(K(\Z,3))$, the homology and the AHSS immediately show that they are purely torsion. Hence, $\text{Hom}_{\Z}( \Omega^{{\rm SO/Spin}}_{9}(K(\Z,3)),\Z) = 0$. This is sufficient to determine $(I_{\Z}\Omega^{\mathrm{SO}/\mathrm{Spin}})^9(K(\Z,3))$, since it receives contributions only from the torsion subgroup of $\Omega^{\mathrm{SO}/\mathrm{Spin}}_{8}(K(\Z,3))$.}
\begin{equation}
\label{eq:invphase-KZ3-physics}
    \begin{tabular}{c|cccccccccc}
  $n= d+2$ & $0$ & $1$ & $2$ & $3$ & $4$ & $5$ & $6$ & $7$ & $8$ & $9$ \\ \midrule
$(I_{\Z} \Omega^{\rm SO})^{n}(K(\Z,3))$ & $\Z$ & $0$ & $0$ & $\Z$ & $\Z$ & $0$ & $\Z_2^2$ & $ \Z$ & $   \Z^2 $ & $  \Z_2^2  $  \\  \midrule
$(I_{\Z} \Omega^{\rm Spin})^{n}(K(\Z,3))$ & $\Z$ & $0$ & $\Z_2$ & $\Z_2\oplus\Z$ & $\Z$ & $0$ & $0$ & $\Z$ & $  \Z^2$ & $  \Z_2  $
\end{tabular}
\end{equation}
When applying the result to anomalies of theories in $d$ spacetime dimensions, one needs to look for invertible phases at the value $n= d+2$. Note that for both $\mathcal S=\mathrm{SO}$ and $\mathcal S=\mathrm{Spin}$, the results in \eqref{eq:tab1informula} show that, for $n\leqslant 4$, the only difference between $\Omega^{\mathcal S}_n(\mathrm{pt})$ and $\Omega^{\mathcal S}_n(K(\Z,3))$ occurs at $n=3$. This reflects the appearance of the Dixmier--Douady class $\iota$ when a $U(1)$-gerbe structure is required.

\subsection{Anomalies of Lower Dimensional Theories ($d\leqslant4$)}

Restricting to $d\leqslant 4$, we identify the corresponding anomalies and discuss, dimension by dimension, both their standard interpretations and some new perspectives. The pure gravitational anomaly parts agree with the corresponding mathematical and physical results, for example \cite{wall1960determination,ABPspinbordismring,Kapustin:2014tfa,Kapustin:2014dxa,Stong}.

\begin{itemize}
    \item $d=0$, $(I_{\Z} \Omega^{\rm SO})^2(K(\Z,3)) = (I_{\Z} \Omega^{\rm SO})^2(\rm pt)) \cong 0$. There is no anomaly from the orientation. When a spin structure is imposed, $(I_{\Z} \Omega^{\rm Spin})^2(K(\Z,3)) = (I_{\Z} \Omega^{\rm Spin})^2(\rm pt)) \cong \Z_2$. There is only a discrete $\Z_2$ anomaly associated with the spin structure (or fermion parity). Concisely, the corresponding bordism group  $\Omega^{\rm Spin}_{1}(K(\Z,3)) \cong  \Omega^{\rm Spin}_1(\rm pt) \cong \Z_2$ is generated by $S^1$ with periodic spin structure (R) and detected by the $\eta$-invariant
    \begin{equation}
        \eta(S_{R}^1) = \frac{1}{2}\,.
    \end{equation}

    \item $d=1$,  $(I_{\Z} \Omega^{\rm SO})^3(K(\Z,3)) \cong \Z$. Consider a 1-dimensional theory with background $U(1)$ gauge field $A$ and with partition function on a closed oriented 1-manifold $M_1$ (finite disjoint union of $S^1$'s) given as
\begin{equation}
    Z_k(M_1,A)=\exp\!\left(2\pi \I \, k\oint_{M_1} A\right),
\qquad k\in \Z\,.
\end{equation}
The theory depends only on the holonomy of $A$, and hence it is topological. Under the 1-form symmetry $A\mapsto A+\lambda$, with $\lambda$ flat, the partition function transforms by $\exp(2\pi \I  k\oint_{M_1}\lambda)$, which does not automatically vanish, because a flat $U(1)$ connection $\lambda$ can still have nontrivial holonomy. Hence, the theory carries a $\Z$-valued 1-form anomaly labeled by $k$.

The ordinary $U(1)$ gauge field $A$ is a connection on a line bundle, whose topological class is classified by $BU(1)=K(\Z,2)$. Since $\Omega^{\rm SO}_1(K(\Z,2))=0$~\cite{Wang:2018qoy}, we may choose an oriented surface $\Sigma_2$ with $\partial\Sigma_2=M_1$ and extend the line bundle with connection over $\Sigma_2$. Then, using the curvature $F=\td A$ of the extended connection, the action can be written as
\begin{equation}
    Z_k(M_1,A)=\exp\!\left(2\pi \I \, k\int_{\Sigma_2}F\right)\,.
\end{equation}
To cancel the variation under the 1-form symmetry, we introduce the inflow term
\begin{equation}
    \exp\!\left(2\pi \I \, k\int_{\Sigma_2} B\right)\,,
\end{equation}
where $B$ is the background 2-form gauge field for the $U(1)$ 1-form symmetry, with $\delta B= -\td \lambda$. From \eqref{eq:tab1informula} we see that $\Omega^{\rm SO}_2(K(\Z,3))=0$, and thus the locally defined expression $\int_{\Sigma_2} B$ can be made well-defined by extending to a 3-manifold  $Y_3$ with $\partial Y_3=\Sigma_2$ and
\begin{equation}
    \exp\!\left(2\pi \I \, k\int_{\Sigma_2} B\right) = \exp\!\left(2\pi \I \, k\int_{Y_3} H\right)\,,
\end{equation}
with the field strength $H=\td B$. In the expression, the cohomology class of $H$ is the pullback of $\iota \in H^3(K(\Z,3),\Z)$ to $H^3(Y_3,\Z)$, and the integer $k$ is precisely the $\Z$-valued anomaly coefficient, matching $(I_{\Z}\Omega^{\mathrm{SO}})^3(K(\Z,3))\cong \Z$. This is the 1-dimensional example of the descent mechanism for 1-form symmetry anomalies.

In the spin case, we have one more $\Z_2$ factor and $(I_{\Z} \Omega^{\rm Spin})^3(K(\Z,3)) \cong \Z_2 \oplus \Z$. This extra $\Z_2$ originates from $\Omega^{\mathrm{Spin}}_2(\mathrm{pt})\cong \Z_2$, and corresponds to the fermion-parity anomaly in one dimension. It is detected by the \emph{Arf invariant} on a closed $2d$ spin manifold, which also determines the phase of the partition function of the Kitaev chain.

\item $d=2$,  $(I_{\Z} \Omega^{\rm SO})^4(K(\Z,3)) \cong \Z \cong (I_{\Z} \Omega^{\rm Spin})^4(K(\Z,3))$, and in both cases one only finds pure gravitational anomalies with the following geometric bases and dual bases (bordism invariants) for the bordism groups $\Omega^{\rm SO}_4(K(\Z,3))  \cong \Omega^{\mathcal{S}}_4(\rm pt) $:
\begin{equation}
\begin{aligned} \label{eq:4dbordism}
    \int_{\mathbb{CP}^2} \frac{1}{3} p_1 &= 1 \qquad \text{for oriented bordism}\,, \\
     \int_{K3} -\frac{1}{48} p_1 &= 1 \qquad \text{for spin bordism}\,.
\end{aligned}
\end{equation}

\item $d=3$,  $(I_{\Z} \Omega^{\rm SO})^5(K(\Z,3)) \cong 0 \cong (I_{\Z} \Omega^{\rm Spin})^5(K(\Z,3))$, and so there is no anomaly.

\item $d=4$, $(I_{\Z} \Omega^{\rm SO})^6(K(\Z,3)) \cong \Z_2^2 $, coming from the torsion part of
\begin{equation}
     \Omega^{\rm SO}_5(K(\Z,3)) \cong \Omega^{\rm SO}_5({\rm pt})  \oplus \tilde  \Omega^{\rm SO}_5(K(\Z,3))   \cong \Z_2 \oplus \Z_2 \,.
\end{equation}
Recall that $\Omega^{\rm SO}_5({\rm pt})$ is generated by the Wu manifold $SU(3)/SO(3)$ and detected by the cup product of Stiefel--Whitney classes
\begin{equation}
    w_2 \smile w_3\,.
\end{equation}
It represents a discrete global gravitational anomaly for $4d$ theories defined on oriented manifolds, without requiring any additional structure. For discussions of this anomaly theory and its $4d$ boundary states, see~\cite{Kapustin:2014tfa,Thorngren:2014pza,Kravec:2014aza,Wang:2018qoy}.

The reduced bordism class gives another $\Z_2$-valued discrete anomaly,
\begin{equation}
    w_2 \smile u \,,
\end{equation}
which is a mixed anomaly between gravity and the $U(1)$ 1-form symmetry. In particular, it obstructs gauging the 1-form symmetry for theories on oriented $4d$ manifolds. A candidate $4d$ boundary theory for this phase is the fermionic monopole phase of Maxwell theory, with coupling
\begin{equation}
   (-1)^{\int_{M_4} w_2\smile \bar c_1}
=
\exp\!\left(\I\pi\int_{M_4} w_2\smile \bar c_1\right) \,,
\end{equation}
 where $\bar c_1$ denotes the mod-2 reduction of the first Chern class of the electromagnetic field. This term forces all monopoles to be fermionic~\cite{Kravec:2014aza,Jian:2020qab}. In the next section, we will derive an analogous statement for magnetic strings in $5d$ theories.  The vanishing
$ (I_{\Z}\Omega^{\mathrm{Spin}})^6(K(\Z,3))\cong 0$
reflects the same phenomenon: on spin manifolds one has $w_2=0$, and thus both the anomaly classes $w_2\smile w_3$ and  $w_2 \smile u$ disappear.

\end{itemize}

\paragraph{$4d$ Maxwell theory on an oriented manifold.}
A special feature of $4d$ Maxwell theory is that the magnetic $U(1)$ symmetry is itself a 1-form symmetry. It is therefore natural to consider the classifying space for both the electric and magnetic $U(1)$ 1-form symmetries, i.e.,
$K(\Z,3)\times K(\Z,3)$.
The oriented bordism groups in $5d$ and $6d$ were computed in~\cite{Davighi:2023luh}, in the context of viewing this as a weak 2-group whose symmetry is purely 1-form:
\begin{equation}
\Omega^{\rm SO}_5(K(\Z,3)\times K(\Z,3)) \cong \Z_2^3\,,
\qquad
\Omega^{\rm SO}_6(K(\Z,3)\times K(\Z,3)) \cong \Z\,.
\end{equation}
Let $\tilde u$ denote the mod-2 reduction of the Dixmier--Douady class $\tilde \iota$ for the background magnetic 2-form gauge field. Compared with $\Omega^{\rm SO}_5(K(\Z,3))$, the extra $\Z_2$-summand in
$\Omega^{\rm SO}_5(K(\Z,3)\times K(\Z,3))$
is precisely detected by
\begin{equation} \label{eq:maganomaly4d}
w_2 \smile \tilde u \,.
\end{equation}
Moreover,
$\Omega^{\rm SO}_6(K(\Z,3)\times K(\Z,3))\cong \Z$
encodes the usual mixed ’t Hooft anomaly between the two $U(1)$ 1-form symmetries with field strengths (fluxes) $H_E$ and $H_M$, described by the anomaly polynomial $\sim H_E\wedge H_M$, in the case of $4d$ Maxwell theory~\cite{Gaiotto:2014kfa}. This is precisely the anomaly obtained in \cite{PartI} from the BV--BRST descent for the two gerbe background fields. Here the bordism computation recovers the same free class and places it in a more complete classification that also includes the torsion part, and hence non-perturbative anomalies.

With the discussions above, we get the well-known four different phases of $4d$ Maxwell theory associated to its total 1-form symmetries:
\begin{equation}
   2\pi \I  \int_{M_6} \iota \smile \tilde \iota \,,\qquad \pi \I  \int_{M_5} w_2 \smile w_3 \,,\qquad\pi \I  \int_{M_5} w_2 \smile u \,,\qquad \pi \I  \int_{M_5} w_2 \smile \tilde u \,.
\end{equation}
These correspond, respectively, to ordinary electrodynamics, all-fermion electrodynamics, electrodynamics with a fermionic monopole, and electrodynamics with a fermionic probe charge~\cite{Jian:2020qab} (see Figure~1 in~\cite{Davighi:2023luh} for a pictorial summary).

For completeness, we also record the corresponding spin bordism groups. Their computation is straightforward, since in these low dimensions the AHSS presents no serious difficulty:
\begin{equation}
\Omega^{\rm Spin}_5(K(\Z,3)\times K(\Z,3)) \cong 0\,,
\qquad
\Omega^{\rm Spin}_6(K(\Z,3)\times K(\Z,3)) \cong \Z\,.
\end{equation}
Thus, the mixed ’t Hooft anomaly persists in the spin case, whereas the vanishing of $w_2$ eliminates the degree-5 $\Z_2$-anomaly.

\paragraph{Pure gravitational anomalies for $d=6$.} Although the case $d=6$ lies in higher dimension, we include it here for completeness. No new anomaly involving the 1-form symmetry appears in this case, as $(I_{\Z} \Omega^{\rm SO/Spin})^{8}(K(\Z,3)) \cong \Z^2$.
Both $\Z$ factors arise entirely from the bordism groups $\Omega^{\rm SO/Spin}_{8}(\rm pt) \cong \Z^2 $, and therefore correspond to ordinary pure gravitational anomalies.

In $8d$, a standard choice of geometric generators for $\Omega^{\mathrm{SO}}_{8}(\mathrm{pt})$ is given by~\cite{MilnorStasheff1974}
\begin{equation}
  \mathbb{CP}^4 \qquad \text{and} \qquad \mathbb{CP}^2 \times \mathbb{CP}^2 \,.
\end{equation}
The dual basis detecting these generators is given by the anomaly polynomials as
\begin{equation}
   \int_{\mathbb{CP}^4} \frac{1}{5}(p_1^2 - 2 p_2)=  \int_{\mathbb{CP}^4} (p_2 - 9 L_2)= 1
   \end{equation}
   and
   \begin{equation}
   \int_{\mathbb{CP}^2 \times \mathbb{CP}^2} \frac{1}{9}(-2p_1^2+5p_2)=\int_{\mathbb{CP}^2 \times \mathbb{CP}^2} ( 10 L_2 - p_2) =1 \,,
\end{equation}
where
\begin{equation}
    L_2=\frac{1}{45}(7p_2-p_1^2)
\end{equation}
denotes the degree-8 part of the Hirzebruch $L$-genus, satisfying
\begin{equation}
   \sigma(M)=\int_{M_8} L_2
\end{equation}
by the signature theorem.

When spin structure is imposed, we have $\Omega^{\mathrm{Spin}}_{8}(\mathrm{pt}) \cong \Z^2$, generated by
\begin{equation}
  \mathbb{HP}^2 \qquad \text{and} \qquad B_8 \,.
\end{equation}
Here $\mathbb{HP}^2$ denotes the quaternionic projective plane, while $B_8$~\cite{Saito:2025idl} is the Bott manifold  satisfying $\int_{B_8}\hat A_2 = 1$ and $p_1=0$.
Recall that the degree-8 part of the $\hat A$-genus
\begin{equation}
\hat A_2=\frac{1}{5760}(7p_1^2-4p_2)\,.
\end{equation}
By the index theorem for the standard Dirac operator $D$, its integral over any closed spin 8-manifold $M_8$ is an integer
\begin{equation}
\text{Ind}(D,M_8)=\int_{M_8} \hat A_2 \in \Z\,.
\end{equation}
We further note that the generators $\mathbb{HP}^2$ and $B_8$ are detected by the ordinary Dirac operator $D$ and by the Rarita--Schwinger operator $D_{\rm RS}$, respectively.

\subsection{$5d$ Theories and a 7-form Anomaly Polynomial}

For $5d$ theories with a $U(1)$ 1-form symmetry, the groups of invertible phases are remarkably torsion free, $(I_{\Z} \Omega^{\rm SO})^{7}(K(\Z,3)) \cong \Z \cong (I_{\Z} \Omega^{\rm Spin})^{7}(K(\Z,3)) $. These are obtained from the bordism groups $ \Omega^{\rm SO}_7(K(\Z,3)) \cong \Z \cong \Omega^{\rm Spin}_7(K(\Z,3))$.

\subsubsection{Geometric Generators and the Dual Basis of the Bordism Groups.} 
We first recall the result for the spin case. Let $(r,s) \in Sp(1) \times Sp(1)$ act on $(A,q) \in Sp(2) \times Sp(1)$ from the right as
\begin{equation}
   \left( A \begin{pmatrix} r & 0 \\ 0 & s \end{pmatrix},\; s^{-1} q s \right).
\end{equation}
Theorem 3.5 of~\cite{Joyce:2025jyd} states that the quotient manifold $X_7 =  \left(Sp(2) \times Sp(1)\right) / \left(Sp(1) \times Sp(1) \right)$ is a closed spin 7-manifold equipped with a $K(\Z,3)$-structure (i.e., a map $f: X_7 \rightarrow K(\Z,3)$), such that its bordism class generates $\Omega^{\rm Spin}_7(K(\Z,3))= \Z$.

Moreover, the normalization of the bordism invariant is such that
\begin{equation} \label{eq:7spinbordisminv}
    \int_{X_7} \frac{1}{4} H_3 \wedge p_1 = 1 \,.
\end{equation}
Here and throughout this subsection, we write $H_3$ for a representative of the class $f^*(\iota)$ and replace cup products by wedge products, following the common physics convention, where $H_3$ denotes the $H$-flux of the $B$-field.

In fact, $X_7$ is another example of an $S^3$-bundle over $S^4$. However, unlike Milnor’s construction $M_{1,1}$ in Section~\ref{sec:Milnor}, this bundle has structure group $SO(3)$, rather than $SO(4)$. Precisely,
\begin{equation}
    X_7=\frac{Sp(2)\times Sp(1)}{Sp(1)\times Sp(1)}
\cong Sp (2)\times_{Sp(1)\times Sp(1)} Sp(1)\,,
\end{equation}
where $(r,s)\in Sp(1)\times Sp(1)$ acts on the fiber $Sp(1)\cong S^3$ via conjugation,
$q\mapsto s^{-1}qs$.
The base manifold is
$Sp(2)/(Sp(1)\times Sp(1))\cong \mathbb{H}\mathbb{P}^1\cong S^4$.

Now we turn to the oriented case, where the normalization of the bordism invariant is fixed as
\begin{equation} \label{eq:7orientedbordisminv}
    H_3 \wedge p_1 \,.
\end{equation}
The reason can be explained by comparing the AHSSs as follows:
\begin{itemize}
    \item In the spin case, from the filtration, one has
\begin{equation}
\text{Gr}\big(\Omega^{\mathrm{Spin}}_7(K(\Z,3))\big) =
G_{3,4}\oplus \frac{G_{5,2}}{G_{3,4}}\oplus \frac{G_{6,1}}{G_{5,2}}\oplus \frac{G_{7,0}}{G_{6,1}}
\cong \Z\oplus \Z_2\oplus \Z_2\oplus \Z_3\,,
\end{equation}
and the total group $ \Omega^{\mathrm{Spin}}_7(K(\Z,3))$ is $\Z$. Hence, if $X$ is a generator of the total spin group, then
$G_{3,4}=12\Z\, X$.
    \item 	In the oriented case, only $E^\infty_{3,4}\cong \Z$ and $E^\infty_{7,0}\cong \Z_3$ occur, so the graded structure is
    \begin{equation}
\text{Gr}\big(\Omega^{\mathrm{SO}}_7(K(\Z,3))\big) =
G_{3,4}\oplus \frac{G_{7,0}}{G_{3,4}}
\cong \Z \oplus \Z_3\,,
\end{equation}
    and if $Y$ is a generator of the total group $\Omega^{\mathrm{SO}}_7(K(\Z,3)) \cong\Z$, then
$G_{3,4}=3\Z\,Y$.

 \item  The forgetful map on the $E^\infty_{3,4}$-piece is induced by
$\Omega_4^{\mathrm{Spin}}(\rm pt)\to \Omega_4^{SO}(\rm pt)$,
which is multiplication by $-16$ [cf.~Eq.~\eqref{eq:4dbordism}].
So
$ 12 X \mapsto -16\times (3 Y)= - 48 Y$,
and hence
$X\mapsto - 4Y$. Note that the minus sign denotes taking the orientation reversal of the manifold.
 \item For $X= X_7$ or $M_{1,1}$, we already have $\int_{X} \frac{1}{4} H_3 \wedge p_1 = 1$, it follows that $\int_{Y}  H_3 \wedge p_1 =- 1.$

\end{itemize}
The next step is to construct a geometric generator $Y$. Consider oriented rank-4 real vector bundles $E$ over $\mathbb{CP}^2$, Grove and Ziller~\cite{Ziller} proved that the isomorphism class is determined by the Stiefel--Whitney class  $ w_2(E)\in H^2(\mathbb{CP}^2,\Z_2)$, the Euler class $e(E) \in H^4(\mathbb{CP}^2,\Z)$, and the Pontryagin class $p_1(E)\in H^4(\mathbb{CP}^2,\Z)$. Let $h \in H^2(\mathbb{CP}^2,\Z) \cong \Z$ be the positive generator so that $H^{4}(\mathbb{CP}^{2},\Z)=\Z\,h^{2}$ and writing $e(E)=m\,h^{2}, p_{1}(E)=n\,h^{2}$.
The allowed values are
\begin{equation}
    \begin{cases}
w_{2}(E)=0\,,\quad n\equiv 2m \pmod 4\,,\\
w_{2}(E)\neq 0\,,\quad n\equiv 2m+1 \pmod 4\,.
\end{cases}
\end{equation}

We pick a specific $E$ with
\begin{equation}
    w_{2}(E)=0\,,\qquad e(E)=0\,,\qquad p_{1}(E)=-4 h^{2} \,,
\end{equation}
and let $Y$ be its associated unit sphere bundle, i.e., we replace the fiber $\R^4$ of $E$ by the unit sphere $S^3$. Then $ Y$ is a closed 7-manifold with $K(\Z,3)$-structure, and $\pi:  Y \rightarrow \mathbb{CP}^2$ is a fibration with fiber $S^3$. By construction, we have
\begin{itemize}
    \item $T Y \oplus \epsilon_1 \cong \pi^*T\mathbb{CP}^2 \oplus \pi^* E$, where $\epsilon_1$ is the trivial $\R$-bundle over $\mathbb{CP}^2$. Then, for $w_2(T\mathbb{CP}^2) \neq 0$, $w_2(T Y)=\pi^*\big(w_2(E)+w_2(T\mathbb{CP}^2)\big) = \pi^*w_2(T\mathbb{CP}^2)\neq 0$. Thus, $Y$ is oriented but non-spin;
    \item the Gysin long exact sequence for $\pi$
\begin{equation*}
   0\;\longrightarrow\; H^{3}( Y,\Z)\;\xrightarrow{\ \pi_!\ }\; H^{0}(\mathbb{CP}^2,\Z)\;\xrightarrow{\ e(E) \wedge \ }\; H^{4}(\mathbb{CP}^2,\Z)\;\xrightarrow{\ \pi^*\ }\; H^{4}( Y,\Z)\;\longrightarrow\; 0\,,
\end{equation*}
    where the map $\pi_!$ is \emph{integration along the fiber} and in our case it is the evaluation of $H_3$ on the fiber $S^3$ which we fix to be 1. The two $0$'s on both ends are just
    $H^{3}(\mathbb{CP}^2,\Z)=0$ and $H^{1}(\mathbb{CP}^2,\Z)=0$. We have the extra condition that $e(E)=0$, and it then follows
    \begin{equation}
    \begin{aligned}
               \pi_!&:H^{3}( Y,\Z)\xrightarrow{\quad\cong \quad} H^{0}(\mathbb{CP}^2,\Z) \cong \Z\,, \\
\pi^*&:H^{4}(\mathbb{CP}^2,\Z)\xrightarrow{\quad\cong\quad}H^{4}( Y,\Z) \cong \Z \,.
    \end{aligned}
    \end{equation}
With the above discussion at hand, we can compute
\begin{equation}
\begin{aligned}
     \int_{ Y}  H_3 \wedge p_1(T Y) &= \int_{ Y} H_3 \wedge \pi^*\bigl(p_1(T\mathbb{CP}^2)+p_1(E)\bigr) \\
&= \int_{ Y} H_3 \wedge\pi^*(3 h ^2-4 h^2)\\
&=- \int_{ Y} H_3 \wedge\pi^*(h^2) = -1\,.
\end{aligned}
\end{equation}
\end{itemize}
Hence, $Y$ is the desired geometric generator of $ \Omega^{\rm SO}_7(K(\Z,3))$ with dual basis $ H_3 \wedge p_1$. This is consistent with the cohomology AHSS computation in Footnote~\ref{footnote:cohomologyAHSS}.

\subsubsection{Green--Schwarz Mechanism and Magnetic Strings}

As discussed in detail in \cite{PartI}, the anomaly polynomial $I_7 = H_3 \wedge p_1$ has the factorized form,\footnote{\label{ft:normalizationofH3p1}In the spin case, for the Green--Schwarz coupling one may take $I_7=H_3\wedge \frac12 p_1$, since $\frac12 p_1$ is an integral class on spin manifolds. This should be distinguished from the primitive spin bordism normalization represented by $\frac14 H_3\wedge p_1$.}
and thus admits a Green--Schwarz mechanism interpretation. We briefly review this below, and then explain an important physical effect of this anomaly in the presence of a magnetic string.

For $5d$ models, we demand them to contain the $U(1)$ (or a $U(1)$ factor of) gauge field $A_1$, and its Wilson lines are charged under the 1-form symmetry. We exclude any dynamical electrically charged particles under $A_1$, such that the $U(1)$ 1-form symmetry is preserved.
The electric global 1-form symmetry acts on $A_1$ by a shift by a flat connection $\lambda_1$:
\begin{equation}
    A_1 \rightarrow A_1 + \lambda^{\text{flat}}_{1} \,,
\end{equation}
with $\td \lambda^{\text{flat}}_{1} = 0$. Gauging this global 1-form means relaxing the flatness condition for $\lambda_1$ by introducing a $U(1)$-gerbe field $B_2$ and replacing the curvature 2-form $F_2=\td A_1$ by the \emph{fake} curvature
\begin{equation}
    \tilde F_2 = F_2 + B_2\,.
\end{equation}
With this modification we have the new Bianchi identity at the kinematic level
\begin{equation}
    \td \tilde F_2 = H_3 = \td B_2\,.
\end{equation}
The fake curvature $\tilde F$ is invariant under the \emph{1-form gauge transformation}
\begin{equation}
\label{eq:1formgaugetrf}
         A_1 \rightarrow A_1 + \lambda_1\,,  \qquad B_2\rightarrow B_2- \td\lambda_1 \,,
\end{equation}
for an arbitrary connection $\lambda_1$. In $5d$, we would also like to invoke a Green--Schwarz type topological coupling (with the conventional notation for the 4-form $X_4 := p_1$ or $\frac{1}{2}p_1$)
\begin{equation} \label{eq:5dGS1}
   \int_{M_5} - A_1 \wedge X_4\qquad \text{or}\qquad
   \int_{M_5} - F_2 \wedge \mathrm{CS}^{\rm grav}_3\,,
\end{equation}
with the properly normalized gravitational Chern--Simons term $\mathrm{CS}^{\rm grav}_3$ such that $\td \mathrm{CS}^{\rm grav}_3 = X_4$.
Alternatively, if there exists a $M_6$ such that $\partial M_6=M_5$, we can express this coupling as a bulk Wess--Zumino term
\begin{equation} \label{eq:6dWZ}
    \int_{M_6} - F_2 \wedge X_4\,,
\end{equation}
which avoids the caveat that $A_1$ or $\mathrm{CS}^{\rm grav}_3$ is defined only locally. In either way, we couple a background magnetic gauge field (whose field strength is $X_4$) to the magnetic 2-form symmetry. As a dynamical consequence, we obtain the equations of motion
\begin{equation}
    \td \star \tilde F_2 = X_4\,.
\end{equation}
We can now consider the variation
\begin{equation}
    \begin{aligned}
       \delta \int_{M_5}- A_1 \wedge X_4 &=   \int_{M_5} -\lambda_1 \wedge X_4 \\
       &=  \int_{M_6} -\td \lambda_1 \wedge X_4 =  \int_{M_6} \delta B_2 \wedge X_4= \delta \int_{M_7} H_3 \wedge X_4 \,,
    \end{aligned}
\end{equation}
where we used the gauge transformation for the gerbe field. Hence, the nontrivial gauge variation of \eqref{eq:5dGS1} under \eqref{eq:1formgaugetrf} corresponds to the anomaly polynomial $I_7$ that we found by bordism computations.

\paragraph{Magnetic strings carry extra topological structure.}
The physical consequence of the bulk coupling~\eqref{eq:6dWZ} can be analyzed in the spirit of~\cite{Wang:2018qoy} (see also discussions in~\cite{Thorngren:2014pza,Cordova:2018acb,Ang:2019txy,Davighi:2023luh}). Concretely, we would like to add to the $5d$ theory a magnetic string, which is the analog of a magnetic monopole in $4d$ and has one spatial dimension.
Let $\Sigma \subset M_5$ be the worldsheet of the magnetic string and we assign magnetic charge 1 to the string. This is achieved by requiring
\begin{equation}
    \td F_2 = 2\pi \delta_{\Sigma}\,,
\end{equation}
where $\delta_{\Sigma}$ is the Poincaré dual to $\Sigma$. In the presence of this magnetic string, we effectively need to excise a tubular neighborhood $N(\Sigma)$ of the string worldsheet from $M_5$ to obtain a 5-dimensional manifold $M_5'$ with boundary $\partial M_5' \neq0$.
Moreover, let $\Theta$ be the unit sphere bundle of the normal bundle of $\Sigma$ inside $M_5$, then $S^2 \rightarrow\Theta \rightarrow \Sigma$ is also a $S^2$-bundle over $\Sigma$ by construction and
\begin{equation}
    \partial M_5' = \Theta\,.
\end{equation}
The magnetic charge of the string is computed as integral on the fiber $S^2$
\begin{equation} \label{eq:magcharge}
    \int_{S^2} \frac{F_2}{2\pi} = 1\,.
\end{equation}

Since we replaced the old closed $5d$ manifold $M_5$ by a new manifold $M_5'$ with boundary, we need to make sense of the couplings~\eqref{eq:5dGS1} and \eqref{eq:6dWZ} with the new data. Although, as mentioned, the two differential form expressions in \eqref{eq:5dGS1} already fail to be globally defined on $M_5$, we can still talk about a genuine $5d$ coupling using \emph{differential cohomology}.\footnote{For physics-oriented accounts of differential cohomology, see~\cite{Cordova:2019jnf,Hsieh:2020jpj,GarciaEtxebarria:2024fuk}. The need for a global definition of TQFT actions on arbitrary spacetime manifolds was already discussed for $3d$ Chern--Simons theory in~\cite{Dijkgraaf:1989pz}, in different terminology.}

Rigorously speaking, we start with two differential cohomology classes: $\check{A}_1 \in \check{H}^2(M_5)$, representing the original $U(1)$ gauge field, and $\check{C}_3 \in \check{H}^4(M_5)$, which serves as the gravitational Chern--Simons field whose field strength is $X_4$. The normalization of $H_3 \wedge p_1$ is essential here, because a differential cohomology refinement can only be defined for an integral class. Hence, one takes $X_4 = p_1$ for oriented non-spin manifolds, and $X_4 = \tfrac{1}{2}p_1$ for spin manifolds, as explained in Footnote~\ref{ft:normalizationofH3p1}.

There is a graded-commutative product on differential cohomology elements. In the present case, we have $\check{A}_1 \cdot \check{C}_3 \in \check{H}^6(M_5) \cong \R / \Z$. Let us denote this on $M_5$ as the pairing
\begin{equation} \label{eq:diffcohopairing}
    ( \check{A}_1, \check{C}_3) := \int_{M_5} \check{A}_1 \cdot \check{C}_3 \in\R / \Z\,,
\end{equation}
which is a well-defined $5d$ quantity.
When all data extends to the bulk $M_6$ with $\partial M_6=M_5$, this pairing is represented by
\begin{equation}
    ( \check{A}_1, \check{C}_3) = \int_{M_6} F\wedge  X_4 \,,
\end{equation}
which agrees with the $6d$ bulk coupling~\eqref{eq:6dWZ} up to a sign. We note that the differential cohomology pairing in~\eqref{eq:diffcohopairing} can be defined directly on a manifold $M$ of general dimensions without introducing the bounding manifold $N$ with $\partial N= M$; see Section~2.3 of~\cite{Okada:2025kie} for the explicit example $M= S^1$.

Under the presence of the magnetic string, the class $\check{A}_1 \cdot \check{C}_3$ now belongs to $\check{H}^6(M_5')$. In order to evaluate this class on the relative fundamental class $[M_5', \partial M_5']$ to obtain a well-defined coupling,  the product $\check{A}_1 \cdot \check{C}_3$ needs to be trivialized on the boundary $\partial M_5'= \Theta$. A natural way to obtain such a trivialization is to trivialize one of the two factors. Recall that the topological classes underlying $\check{A}_1$ and $\check{C}_3$ are
\begin{equation}
    c_1 \qquad \text{and }\qquad X_4\,,
\end{equation}
where $c_1 = \left[\frac{F_2}{2\pi}\right]$ is the first Chern class of the gauge bundle on $M_5'$.
The magnetic charge quantization condition~\eqref{eq:magcharge} forces $c_1$ to be nontrivial on $\Theta$; hence, we need a trivialization of $X_4$ on $\Theta$, which can be represented locally by
\begin{equation}
    \td  \psi_3 = X_4 \,,
\end{equation}
for some three form $\psi_3$.
In the present situation this existence condition is automatic.
Indeed, $\Theta$ is the $S^2$-bundle over the magnetic string worldsheet $\Sigma$, and for any such 4-manifold one has
\begin{equation}
    p_1(T\Theta)=0 \,.
\end{equation}
Hence, $X_4|_\Theta$ vanishes in cohomology both in the oriented case, where $X_4=p_1$, and in the spin case, where $X_4=\frac12 p_1$. Therefore, a trivialization of $X_4$ on $\Theta$ always exists.

However, the choices of $\psi_3$ are additional data. More precisely, the set of trivializations is not canonical; rather, it forms a \emph{torsor}\footnote{It means different choices differ by a class of $H^3(\Theta,\Z)$.} for $H^3(\Theta,\Z)$. By Gysin sequence, one has
\begin{equation}
    H^3(\Theta,\Z)\cong H^1(\Sigma,\Z)\,.
\end{equation}

\begin{itemize}
    \item In the oriented case, $X_4=p_1$. Thus, after choosing one reference trivialization $\psi^0_3$, any other one is obtained by shifting by a class in $H^1(\Sigma,\Z)$.

    \item  In the spin case, $X_4= \frac{1}{2} p_1$. A trivialization of $\frac{1}{2}p_1$ is called a \emph{string structure}~\cite{Tachikawa:2021mvw}; in our case, the set of string structures forms a torsor for $H^1(\Sigma,\Z)$.

    \item Finally, the spin trivialization is a refinement of the oriented one: forgetting the spin refinement sends a trivialization of $\frac12 p_1$ to a trivialization of $p_1$, and under the $H^3(\Theta,\Z)$-action this map is given by multiplication by $2$. Equivalently, once a reference spin trivialization is chosen, the corresponding oriented trivializations differ only by \emph{even} classes in $H^1(\Sigma,\Z)$.
\end{itemize}

We thus conclude that, in the $H_3 \wedge X_4$ phase of $U(1)$ 1-form symmetry, the magnetic strings are not specified solely by their magnetic charge and embedding.
They also carry an additional topological sector, corresponding to a choice of trivialization of the boundary class on the sphere bundle $\Theta$.
These sectors are classified by $H^3(\Theta,\Z)\cong H^1(\Sigma,\Z)$.
Equivalently, $H^1(\Sigma,\Z)$ labels the distinct topological sectors of the magnetic string worldsheet.
In the spin case, the corresponding sectors refine the oriented ones, and map to them by multiplication by $2$.

\subsubsection{$(I_{\Z} \tilde \Omega^{\rm SO/Spin})^{7}(K(\Z,3) \times K(\Z,4))$ and Phases of $5d$ Maxwell Theory}
We briefly reviewed the phases of $U(1) \times U(1)$ 1-form symmetries for $4d$ Maxwell theory earlier in this section. With $(I_{\Z} \Omega^{\rm SO})^{7}(K(\Z,3)) \cong \Z \cong (I_{\Z} \Omega^{\rm Spin})^{7}(K(\Z,3))$ at hand, we now explore phases of $5d$ Maxwell theory. For this reason, we should likewise include the magnetic dual symmetry background, a $U(1)$ 2-form symmetry whose 3-form background gauge field is classified topologically by $K(\Z,4)$. In other words, we need to look at $(I_{\Z} \Omega^{\rm SO/Spin})^{7}(Q)$ for $Q:= K(\Z,3) \times K(\Z,4)$. Let $\gamma \in H^4(K(\Z,4),\Z)$ be the generator, and let $\bar \gamma$ denote its mod-2 reduction. In Appendix~\ref{sec:OmegaKZ3KZ4}, we computed $\tilde \Omega^{\rm SO/Spin}_{6,7}(Q)$, and the result and corresponding invertible phases read
\begin{equation}
    \begin{tabular}{c|cc}
  $n$ & $6$ & $7$  \\ \midrule
$\tilde \Omega^{\rm SO}_n(Q)$ & $\Z_2$ & $\Z^2$ \\  \midrule
$\tilde \Omega^{\rm Spin}_n(Q)$ & $0$ & $\Z^2$
\end{tabular}
\qquad \qquad
    \begin{aligned}
(I_{\Z} \tilde \Omega^{\rm SO})^{7}(Q) &\cong \Z \,\iota\, \gamma \oplus \Z \,\iota \, p_1  \oplus \Z_2 \, w_2 \,\bar \gamma\,,   \\ 
(I_{\Z} \tilde \Omega^{\rm Spin})^{7}(Q) &\cong \Z \,\iota\, \gamma \oplus \Z \,\frac{1}{4}\iota \, p_1   \,.
\end{aligned}
\end{equation}
$\iota \smile p_1$ is just the integral lift of our anomaly polynomial $H_3 \wedge p_1$. We also see the expected mixed 't Hooft anomaly between the electric $U(1)$ 1-form and the magnetic $U(1)$ 2-form symmetries:
\begin{equation}
    \iota \smile \gamma\,,
\end{equation}
or in terms of differential forms:
\begin{equation}
    H_3 \wedge H_4\,,
\end{equation}
where $H_4$ is the 4-form field strength of the magnetic 3-form gauge field whose integral flux is $\gamma$.

\paragraph{The $\Z_2$ valued 6-cocycle $w_2 \smile \bar \gamma$.} An observation is that, by adding the magnetic dual, we have the $\Z_2$-anomaly on oriented manifolds:
\begin{equation}
    w_2 \smile \bar \gamma\,.
\end{equation}
This is the analog of~\eqref{eq:maganomaly4d}. From anomaly inflow perspective it cancels the anomaly from the topological coupling
\begin{equation} \label{eq:anomaly5ddd}
     \pi \I\int_{M_5} w_2 \smile \bar{d}^{\rm dual}_3 \,.
\end{equation}
Here, $\bar{d}^{\rm dual}_3$ is the mod-2 reduction of the characteristic class of the magnetic dual 3-form gauge field. The presence of~\eqref{eq:anomaly5ddd} will force all electrically charged particles to be fermions, exactly as in the $4d$ case. On spin manifolds it trivializes, as $w_2=0$ is forced by spin structure.

\subsection{$7d$ Theories and Discrete Anomalies}

For $d=7$, we can read off the invertible phases from \eqref{eq:invphase-KZ3-physics} at $n=d+2 = 9$ and
\begin{equation}
    \begin{aligned}
(I_{\Z}  \Omega^{\rm SO})^{9}(K(\Z,3))& \cong  \Z_2 \, u \,\mathrm{Sq}^2 u   \oplus \Z_2 \, u \,w_2 w_3  \,,\\  (I_{\Z}  \Omega^{\rm Spin})^{9}(K(\Z,3)) &\cong \Z_2 \, u \,\mathrm{Sq}^2 u  \,.
\end{aligned}
\end{equation}
The contribution comes purely from the torsion part of $\Omega^{\rm SO/Spin}_8(K(\Z,3))$ and this leads to a $\Z_2$-valued discrete anomaly inflow from the $8d$ bulk. If we look at a $7d$ theory defined on an oriented but non-spin manifold $M_7$, then there are two possible $\Z_2$-valued anomalies:
\begin{equation} \label{eq:7dorientedphases}
 u \smile \mathrm{Sq}^2 u \qquad\text{and} \qquad u \smile w_2 \smile w_3\,.
\end{equation}
We assume here that we have a $U(1)$ (or a $U(1)$ factor) gauge field $A_1$ with first Chern class $c_1$, whose Wilson lines are charged under the 1-form symmetry.

\paragraph{The pure $U(1)$ 1-form anomaly $u\smile \mathrm{Sq}^2 u$.} For the anomaly $u\smile \mathrm{Sq}^2 u$, there is a direct analogy with the 4-dimensional axionic Green--Schwarz coupling $\phi\,F\wedge F$, which descends from the six-form anomaly polynomial $F^3$. Namely, once one chooses a boundary St\"uckelberg type relation 
\begin{equation}
    u=\delta \bar c_1\,,
\end{equation}
where $\bar c_1$ is understood as the mod-$2$ first Chern class in the lifted boundary description, the 8-dimensional higher-form anomaly $u\smile \mathrm{Sq}^2 u$ admits the $7d$ descendant
\begin{equation}
    \bar c_1\smile \mathrm{Sq}^2 u\,.
\end{equation}
This, however, requires introducing the gerbe field $B_2$ into the $7d$ description and effectively corresponds to a Higgsed realization of its 1-form gauge symmetry on the boundary. Now this topological term $u\smile \mathrm{Sq}^2 u$ appears as a $U(1)$ 1-form gauge anomaly of a $7d$ theory containing the dynamical pair $(A_1, B_2)$. We interpret the presence or absence of such anomaly of the $U(1)$ 1-form symmetry, as a choice of new discrete $\theta$-angle $(\theta=0,\pi)$, in the definition of the $7d$ generalized Maxwell theory.

\paragraph{Anomaly interplay.} Alternatively, since the 8-cocycle $u\,\mathrm{Sq}^2 u$ is written entirely in terms of $\Z_2$-valued classes, one might wonder whether it should instead be interpreted as an anomaly of the discrete $\Z_2\subset U(1)$ 1-form subgroup, so that the corresponding boundary theory could be described purely in terms of $\Z_2$-cocycles on $K(\Z_2,2)$, in the sense of~\cite{Kapustin:2014zva}.

More generally, the embedding of anomalies of a subgroup into those of an ambient symmetry group was studied long ago in~\cite{Elitzur:1984kr}, and is often referred to nowadays as \emph{anomaly interplay}. For example, the $4d$ $SU(2)$ Witten anomaly can be related to the perturbative anomaly of $U(2)$~\cite{Davighi:2020bvi}. An example involving anomaly interplay for 2-groups may be found in~\cite{Davighi:2023luh}. Our discussion here proceeds in the reverse direction: we ask how a $\Z_2$-valued anomaly of the $U(1)$ one-form symmetry behaves when restricted to its $\Z_2$ subgroup.

The inclusion $i:\Z_2\hookrightarrow U(1)$ induces a map
\begin{equation}
    B^2i:K(\Z_2,2)\longrightarrow K(\Z,3)\,,
\end{equation}
and hence a pullback isomorphism
\begin{equation}
(B^2i)^*:H^3(K(\Z,3),\Z_2)\longrightarrow H^3(K(\Z_2,2),\Z_2)\,.
\end{equation}
If $x$ denotes the generator of $H^2(K(\Z_2,2),\Z_2) \cong \Z_2$, then
$(B^2i)^*(u)=\mathrm{Sq}^1 x.$
It follows that
\begin{equation}
\label{pullbackanomaly}
    (B^2i)^*(u\smile \mathrm{Sq}^2 u)
    = \mathrm{Sq}^1 x \smile \mathrm{Sq}^2(\mathrm{Sq}^1 x)
\,,
\end{equation}
which is an \emph{admissible} polynomial generator in the mod-2 cohomology ring of $K(\Z_2,2)$. Thus, $u\smile \mathrm{Sq}^2 u$ pulls back nontrivially to $K(\Z_2,2)$, and therefore we should look at the result of $\tilde \Omega^{\rm SO/Spin}_8(K(\Z_2,2))$ to determine whether this anomaly trivializes after restriction to the subgroup $\Z_2\subset U(1)$. For the spin case, the reduced bordism group is $\tilde \Omega^{\rm Spin}_8(K(\Z_2,2))\cong \Z_8$.\footnote{$\tilde \Omega^{\rm Spin}_n(K(\Z_2,2))$ can be directly obtained from $ko_n(K(\Z_2,2))$ computed using the Adams spectral sequence in \cite{davis2025connective}.} The mod-$2$ reduction of a generator of $E^\infty_{8,0}\cong\Z_8$ is characterized by annihilating both the image of $\mathrm{Sq}^1:H^7(K(\Z_2,2),\Z_2)\to H^8(K(\Z_2,2),\Z_2)$ and the image of $\mathrm{Sq}^2:H^6(K(\Z_2,2),\Z_2)\to H^8(K(\Z_2,2),\Z_2)$. The two relevant images are generated by $\mathrm{Sq}^1x\,\mathrm{Sq}^2\mathrm{Sq}^1x+x(\mathrm{Sq}^1x)^2$ and $x^4+x(\mathrm{Sq}^1x)^2$. Since \eqref{pullbackanomaly} is not in their span, it pairs nontrivially with the $E^\infty_{8,0}$ class. Therefore, the anomaly of $U(1)$ 1-form symmetry does not trivialize after restriction to $\Z_2\subset U(1)$. Instead, it becomes an order-two element inside the intrinsic $\Z_8$ anomaly group of the $\Z_2$ 1-form symmetry.

The latter anomaly of~\eqref{eq:7dorientedphases} corresponds to the $7d$ topological coupling
\begin{equation} \label{eq:7duw2w3}
    \pi \I  \int_{M_7} \bar {c}_1 \smile w_2 \smile w_3\,.
\end{equation}
Assuming the presence of~\eqref{eq:7duw2w3}, one can ask again what is the physics consequence here. The interpretation is straightforward; we proceed in a similar manner to the previous section.

\paragraph{Magnetic branes on oriented manifolds.}
In $7d$, the magnetic dual of a $U(1)$ 1-form gauge field $A_1$ is a $U(1)$ 4-form gauge field. Accordingly, the objects magnetically charged under $A_1$ are 3-dimensional extended objects, which we will refer to as magnetic branes. Let $W$ denote the 4-dimensional world volume of such a magnetic brane, and suppose that
\begin{equation}
    \td F_2 = 2\pi \,\delta_W \,,
\end{equation}
so that the magnetic brane carries unit magnetic charge.

In the presence of this magnetic brane background, we would like the coupling~\eqref{eq:7duw2w3} to be well defined on
\begin{equation}
   M_7' := M_7 \setminus N(W)\,,
\end{equation}
where $N(W)$ is a tubular neighborhood of $W$. For the coupling~\eqref{eq:7duw2w3} to be well defined on $M_7'$, one must specify a trivialization of the class $\bar c_1 \smile w_2 \smile w_3$ along the boundary $\partial M_7'$.

The boundary $\partial M_7'$ is the $S^2$-bundle over $W$ associated to the normal bundle of the embedding $W\subset M_7$. Since the magnetic charge is nontrivial, $\bar c_1$ restricts nontrivially to the $S^2$ fiber of $\partial M_7'$. Therefore, one cannot trivialize $\bar c_1$ on the boundary. In order to define the coupling, one must instead trivialize either $w_2$ or $w_3$ along $\partial M_7'$.
\begin{itemize}
    \item A trivialization of $w_2$ is precisely a spin structure on $\partial M_7'$.
    \item On the other hand, a trivialization of $w_3$ can also make the coupling well-defined. We emphasize that here we mean a trivialization of the mod-2 class $w_3$, not of its integral lift $W_3$ (which is relevant for anomaly cancellation in string theory with $D$-branes~\cite{Freed:1999vc}); thus, the relevant boundary datum is not a $\mathrm{Spin}^c$ structure, but rather a choice of null-homotopy for $w_3$.

\end{itemize}

Thus, in the magnetic brane background, the coupling~\eqref{eq:7duw2w3} can be defined only after making one of these additional choices on the boundary.  Unlike in the $4d$ and $5d$ cases, the relevant structure is not naturally an intrinsic property of the brane world volume $W$ alone. Rather, the world volume should be viewed as carrying a fermionic structure twisted by the transverse bundle. In other words, the brane supports fermionic, or more generally sign-sensitive, data that is defined only relative to the way $W$ is embedded in the ambient spacetime $M_7$, and not solely in terms of the intrinsic geometry of $W$.

If our $7d$ theory is defined on a spin manifold, then $w_2 w_3$ vanishes, and hence we are left only with the $\Z_2$-invertible phase $ u\,\mathrm{Sq}^2 u$.

\paragraph{$7d$ Maxwell and its phases.}
For a toy model, we can again take $7d$ free Maxwell theory. This amounts to include the background 5-form gauge field for the magnetic $U(1)$ 4-form symmetry, which is classified by $K(\Z,6)$. Let $T= K(\Z,3) \times K(\Z,6)$, its associated invertible phases are computed in Appendix~\ref{sec:OmegaKZ3KZ6}. We denote by $H_6$ the field strength of the magnetic 5-form background gauge field. Its integral flux is the generator $\zeta \in H^6(K(\Z,6),\Z) \cong \Z$, with mod-2 reduction $\bar \zeta$. The results are
\begin{equation}
    (I_{\Z} \tilde \Omega^{\rm SO})^{9}(K(\Z,3) \times K(\Z,6)) \cong \Z \,\iota \zeta \oplus \Z_2 \, u\,\mathrm{Sq}^2 u \oplus \Z_2 \,u \ w_2 \ w_3 \oplus \Z_2 \,\bar \zeta w_2  \,,
\end{equation}
and
\begin{equation}
   (I_{\Z} \tilde \Omega^{\rm Spin})^{9}(K(\Z,3) \times K(\Z,6)) \cong \Z \,\iota \zeta \oplus \Z_2 \, u\,\mathrm{Sq}^2 u   \,.
\end{equation}
In both cases, we have the mixed 't Hooft anomaly $\iota \smile \zeta$ or just $H_3 \wedge H_6$. A new $\Z_2$ anomaly appears in the oriented setup
\begin{equation}
     w_2 \smile\bar \zeta \,,
\end{equation}
which is again the mixed 't Hooft anomaly between the magnetic $U(1)$ 4-form symmetry and diffeomorphism symmetry. In such a phase, the electrically charged particles are all fermions.

\section{Top-Down Constructions of Anomalies for $U(1)$ 1-Form Symmetries}
\label{sec:topdown}

In this section, we would like to realize the new anomalies for $U(1)$ 1-form symmetry, with top-down string/M-theory constructions, following the lines of \cite{Apruzzi:2021nmk,vanBeest:2022fss}.

\subsection{$H_3\wedge p_1$ Anomaly for $5d$ Theories}

We first construct a physical example realizing the 7-form anomaly polynomial $H_3\wedge p_1$ for $5d$ quantum field theories.

We consider the IIA superstring theory on $\mc{M}_{10}=\mc{M}_6\times L_4$, where $L_4$ is a compact 4-manifold. In IIA superstring theory, there is a topological action
\be
\ba
I_{10}&=-\frac{1}{2}B_2\wedge G_4\wedge G_4-B_2\wedge X_8\,,\cr
X_8&=\frac{1}{192}(p_1(\mc{M}_{10})^2-4p_2(\mc{M}_{10}))\,,
\ea
\ee
with the gauge invariant 11-form
\be
\label{5d-I11}
I_{11}=-\frac{1}{2}H_3\wedge G_4\wedge G_4-H_3\wedge X_8\,.
\ee
Here $\mathcal M_6$ is viewed as the spacetime of the anomaly theory whose boundary supports the $5d$ QFT. Dimensionally reducing $I_{11}$ over $L_4$ produces a 7-form anomaly polynomial containing $H_3\wedge p_1$. In this setup, the NS-NS 2-form gauge field $B_2$ is identified with the background gauge field of the $U(1)$ 1-form symmetry in $5d$, through the expansion over $1\in H^0(L_4)=\Z$.

For the R-R 4-form flux $G_4$, we should be careful that it is quantized as~\cite{Witten:1996md}
\be
\label{G4-quantization}
G_4=a-\frac{p_1(T\mc{M}_{10})}{4}\,,
\ee
where $a\in H^4(\mc{M}_{10},\Z)$ is an integral 4-cocycle. Consequently, the first term in \eqref{5d-I11} also contributes to the term $H_3\wedge p_1$. As we only focus on the $U(1)$ 1-form symmetry associated with $H_3$, here we do not consider other background gauge fields from the expansion of $a$.

Using the K\"unneth formula we expand
\begin{align}
p_1(T\mc{M}_{10})&=p_1(T\mc{M}_6)+p_1(TL_4)\,,\\
p_2(T\mc{M}_{10})&=p_2(T\mc{M}_6)+p_2(TL_4)+p_1(T\mc{M}_6) p_1(TL_4)\,.
\end{align}
Hence, after we reduce $I_{11}$ over $L_4$, we get the term
\be
I_7=-\frac{5}{96}\bigg(\int_{L_4}p_1(TL_4)\bigg) (H_3\wedge p_1(T\mc{M}_6))\,.
\ee
Using the Hirzebruch signature theorem $\int_{L_4}p_1(TL_4)=3\sigma(L_4)$ where $\sigma(L_4)$ is the signature of $L_4$, we obtain
\be
\label{SUGRA-I7}
I_7=-\frac{5}{32}\sigma(L_4) (H_3\wedge p_1(T\mc{M}_6))\,.
\ee

In particular, when $L_4=S^4$, the $5d$ theory is the reduction of $6d$ $\mathfrak{u}(N)$ (2,0) theory on $S^1$, and hence the $5d$ $U(N)$ maximal SYM. The $U(1)$ 1-form symmetry is then identified as the center symmetry of the $5d$ $U(N)$ maximal SYM. However, there is no anomaly for such symmetry in this case, as $\sigma(S^4)=0$.

The anomaly arises whenever $\sigma(L_4)\neq 0$ $(\mathrm{mod}\ 32)$, i.e., $L_4=\mb{CP}^2$ or a del Pezzo surface $dP_n$ $n>1$. However, in such cases $L_4$ is not a spin manifold unless $\sigma(L_4)=0$ $(\mathrm{mod}\ 16)$. When $L_4$ is not a spin manifold, the SUGRA solution does not exist. The non-supersymmetric background of IIA on $AdS_6\times L_4$ may exist at a finite string coupling, similar to the cases of non-supersymmetric orbifolds in \cite{Braeger:2024jcj,Heckman:2024zdo}.

If we want $L_4$ to be a spin manifold, then we can take $L_4=K3$ with $\sigma(K3)=-16$. However, as the $K3$ surface is Ricci-flat, this cannot define a valid $AdS_6\times L_4$ background, either.

\subsection{$u\, \mathrm{Sq}^2 u$ Anomaly for $7d$ Theories}

Now we investigate the possibility of realizing the anomaly term $u\,\mathrm{Sq}^2 u$ for a $U(1)$ 1-form symmetry in a $7d$ QFT, constructed from a top-down approach. As a hint, the $10d$ IIA superstring theory has a topological term $I_{10}$ defined mod 2~\cite{Diaconescu:2000wy}, satisfying the following identities over a 10-manifold $\mc{M}_{10}$:
\begin{align}
I_{10}&= \bar a \smile \omega_6= \bar a\smile \mathrm{Sq}^2 \bar a\,,
\end{align}
where $\bar a$ is the mod 2 restriction of the integral cohomology class (\ref{G4-quantization}), $\bar a\in H^4(\mc{M}_{10},\Z_2)$. Note that the normalization is chosen such that $\int_{M_{10}}2I_{10}\in\Z$.

Let us consider IIA superstring theory on $\mc{M}_8\times L_2$, and we reduce $\bar a$ over a $\Z_2$ torsional cohomology class $\theta\in H^1(L_2,\Z_2)$ as
\be
\bar a=\theta \smile u\,.
\ee
With the K\"unneth formula for the Steenrod square $\mathrm{Sq}^2$, we derive
\be
\ba
I_8=\int_{L_2} I_{10}=\int_{L_2}(\theta\smile u)(\mathrm{Sq}^2 u\smile\theta)=\bigg(\int_{L_2}\theta\smile \theta\bigg) u\,\mathrm{Sq}^2 u\,.
\ea
\ee
Hence, the anomaly term of the form $u\,\mathrm{Sq}^2 u$ arises if we choose a compact manifold $L_2$ with the pairing
\be
\int_{L_2}\theta\smile \theta=1
\ee
for the $\Z_2$-valued 1-cocycle $\theta$. Such example of $L_2$ can be chosen as the real projective plane $\mathbb{RP}^2$.

Nonetheless, in this case we should typically interpret $\Z_2$-valued 3-cocycle $u$ as the background gauge field for a $\Z_2$ 2-form symmetry, instead of the field strength for a $U(1)$ 1-form symmetry. Thus, despite the top-down construction of the anomaly of the expected form $u\,\mathrm{Sq}^2 u$, we do not directly obtain $u$ as the $\Z_2$ restriction of the field strength for $U(1)$ 1-form symmetry.

\section{Conclusions}
\label{sec:conclusion}
In this work, we studied anomalies of $U(1)$ 1-form symmetries from the viewpoint of invertible phases and bordism. The main technical input was the computation of $\Omega^{\rm SO/Spin}_{d\leqslant 8}(K(\Z,3))$, together with the identification of the bordism invariants and geometric generators. The free part of the Anderson dual groups reproduces perturbative anomaly polynomials, while the torsion part gives global anomalies invisible to the local descent formalism. This distinction is manifested clearly in the anomalies of $U(1)$ 1-form symmetry: in $5d$ the anomaly is detected by the local class $H_3\wedge p_1$, whereas in $7d$ it is detected by the genuinely $\Z_2$-valued bordism invariant $u\,\mathrm{Sq}^2 u$.

The anomaly polynomial $H_3\wedge p_1$ gives a mixed anomaly between the $U(1)$ 1-form symmetry and spacetime diffeomorphisms. Its physical effect can be seen in the magnetic string sector: in the presence of the anomaly, a magnetic string is specified not only by its charge and embedding, but also by a choice of trivialization of the relevant characteristic class on the boundary of a tubular neighborhood. This is the direct higher-form analogue of the way fermionic monopoles appear in $4d$ Maxwell theory. In the spin case, the characteristic class is refined from $p_1$ to $\frac12 p_1$; forgetting this refinement maps the corresponding trivializations to the oriented ones by multiplication by 2.

The torsion class $u\,\mathrm{Sq}^2 u$ is more subtle. We interpreted it as a discrete $\theta$-angle for a generalized Maxwell-type theory, after imposing a boundary St\"uckelberg-type relation between the dynamical field $A_1$ and the 2-form field $B_2$. This provides a useful boundary realization of the anomaly, but not yet an intrinsic description in which $B_2$ remains a purely background field for the $U(1)$ 1-form symmetry. Finding such an intrinsic boundary construction would be an important next step. We also found that the class does not become trivial after restriction to the subgroup $\Z_2\subset U(1)$; rather, it maps to an order-two element in the $\Z_8$ anomaly group of a $\Z_2$ 1-form symmetry. This gives a concrete example of anomaly interplay for higher-form symmetries in the direction from a continuous symmetry to a finite subgroup.

Several extensions are natural. First, one can replace the pure $K(\Z,3)$ background by a higher-group background in which the $B$-field is only one component of a larger symmetry structure. In that case the relevant target space is no longer simply a product of Eilenberg--Mac Lane spaces, and the Postnikov data should affect both the AHSS differentials and the possible anomaly terms. It would also be interesting to understand whether analogous structures exist for non-Abelian gerbes or other models of non-Abelian higher gauge fields, although such backgrounds are not classified by $K(\Z,3)$ and require a different input.

Second, one can consider higher-form analogues. For a continuous $U(1)$ $p$-form symmetry, the topological background is classified by $K(\Z,p+2)$. The same bordism strategy should apply, but the relevant Steenrod operations and extension problems will change with $p$. Even for $K(\Z,3)$ itself, it would be useful to continue the computation beyond degree 8 in order to see how the higher torsion classes enter the anomaly classification in dimensions greater than seven.

Finally, throughout this work we considered oriented (bosonic) and spin (fermionic) theories without time-reversal symmetry. Including time reversal would require more general tangential structures, such as $\mathrm O$, $\mathrm{Pin}^{\pm}$, etc., and may lead to additional bordism invariants and new mixed anomalies between $U(1)$ $1$-form symmetries and spacetime symmetries. These directions should help clarify how continuous higher-form symmetries fit into the general classification of invertible phases and anomaly inflow.

\section*{Acknowledgements}
We would like to thank Hank Chen, Jingyuan Chen, Zheng-Cheng Gu, Marc Klinger, Tian Lan, Shota Saito, Yuji Tachikawa, Qing-Rui Wang, Weicheng Ye, Piljin Yi, and Hao Y.~Zhang for discussions. WJ is supported by the Research Grants Council of Hong Kong Special Administrative Region of China (Project No.~14302725) under the scheme of General Research Fund. YNW is supported by National Natural Science Foundation of China under Grant No.~12422503. YZ is supported by WPI Initiative, MEXT, Japan at Kavli IPMU, the University of Tokyo.

\appendix

\section{Computation of $ \tilde \Omega^{\rm SO/Spin}_{6,7}(K(\Z,3) \times K(\Z,4))$}
\label{sec:OmegaKZ3KZ4}
The reduced (co)homology data (denoted by $\tilde H$) of $Q=K(\Z,3) \times K(\Z,4)$ can be computed from those of $K(\Z,3)$ and $K(\Z,4)$ by the K\"unneth formula. The results are given in Table \ref{tab:Qdata}.
\begin{table}[H]
\centering
    \begin{tabular}{c|cccccc}
  $p$ &  $3$ & $4$ & $5$ & $6$ & $7$ & $8$  \\ \midrule
$\tilde H^p(Q,\Z)$  & $\Z \iota$ & $ \Z \gamma $ & $0$ & $\Z_2 \iota^2$ & $\Z \iota \gamma \oplus \Z_2 \xi$ & $ \Z \gamma^2 \oplus \Z_3 \kappa $ \\ \midrule
$\tilde H_p(Q,\Z)$   & $\Z$ & $ \Z $ & $\Z_2$ & $\Z_2$ & $\Z \oplus \Z_3$ & $ \Z \oplus \Z_2 \oplus \Z_3 $ \\ \midrule
$\tilde H^p(Q,\Z_2)$ & $\Z_2 u $ & $ \Z_2 \bar \gamma $ & $\Z_2 \, \mathrm{Sq}^2 u$ & $\Z_2 u^2 \oplus \Z_2 \,\mathrm{Sq}^2 \bar \gamma$ & $\Z_2 u \bar \gamma \oplus \Z_2 \,\mathrm{Sq}^3 \bar \gamma$ & $\Z_2 \bar \gamma^2\oplus \Z_2 u\,\mathrm{Sq}^2 u$ \\ \midrule
$\tilde H_p(Q,\Z_2)$  & $\Z_2$ & $ \Z_2 $ & $\Z_2$ & $\Z_2^2$ & $\Z_2^2$ & $\Z_2^2$
\end{tabular}
\caption{(Co)homology data of $Q=K(\Z,3) \times K(\Z,4)$.}
\label{tab:Qdata}
\end{table}
We follow the conventions of Section~\ref{sec:KZ3}. Here
$\gamma$ and $\xi$ denote the generators of
$H^4(K(\Z,4),\Z)\cong \Z$ and
$H^7(K(\Z,4),\Z)\cong \Z_2$, respectively. The nontrivial groups start from $p=3$ and barred letters represent mod-2 reduction of the corresponding integral classes. There is an extra relation $\bar \xi = \mathrm{Sq}^3 \bar \gamma$. We will use $(I_{\Z}\tilde \Omega^{\mathcal S})^n(X)$ to denote the reduced contribution to $(I_{\Z}\Omega^{\mathcal S})^n(X)$ computed from $\tilde \Omega^{\mathcal S}_\bullet(X)$.

\paragraph{Oriented case.}
Since we are only targeting $\tilde \Omega^{\rm SO/Spin}_{6,7}(Q)$, we only show the relevant nonvanishing entries in Table~\ref{tab:E2orientedQ}. Here $C=\tilde H_7(Q,\Z) \cong \Z \oplus \Z_3$ and $D=\tilde H_8(Q,\Z) \cong \Z \oplus \Z_2 \oplus \Z_3$. An immediate result is that
\begin{equation}
    \boxed{\tilde \Omega^{\rm SO}_6(Q) \cong  \tilde H_6(Q,\Z)\cong \Z_2}\,.
\end{equation}
\begin{sseqdata}[ name = E2orientedQ, classes = { draw = none }, axes type =  frame, scale = 0.8 ]
\class(0,0)
\class(0,7)
\class["\Z_2"](6,0)
\class["C"](7,0)
\class["D"](8,0)
\class["\Z"](3,4)
\end{sseqdata}
\begin{table}[htbp]
  \centering
  \printpage[
    name = E2orientedQ,
    grid = chess,
  ]
  \caption{The relevant $E^2$-page entries for $E^2_{p,q} = H_p(Q,\Omega^{\rm SO}_q({\rm pt}))$.}
  \label{tab:E2orientedQ}
\end{table}
Moreover, by the K\"unneth formula, the only degree-6 contribution comes from the $K(\Z,4)$ factor of $Q$, and hence
$H_6(Q,\Z)\cong H_6(K(\Z,4),\Z)\cong \Z_2$. We see that this 2-torsion group is detected by the generator $\mathrm{Sq}^2 \bar \gamma \in H^6(K(\Z,4),\Z_2) \cong \tilde H^6(K(\Z,4),\Z_2)$, and for any class $[M_6,f] \in \tilde \Omega^{\rm SO}_6(Q)$ we have the Wu formula
\begin{equation}
    \int_{M_6} \mathrm{Sq}^2 \bar \gamma = \int_{M_6} w_2 \smile \bar \gamma \in \Z_2 \,.
\end{equation}

Let $p: Q \rightarrow K(\Z,3)$ be the projection map, then, it induces a natural map between the AHSSs of $Q$ and of $K(\Z,3)$ such that the following particular diagram is commutative
\begin{equation}
\begin{tikzcd} \label{eq:commdiaQ}
 E^5_{8,0}(Q) \cong  \tilde H_8(Q,\Z)  \arrow[r, "p_*"] \arrow[d," \td_Q^5"'] & E^5_{8,0}(K(\Z,3)) \cong \tilde H_8(K(\Z,3),\Z) \cong \Z_2 \arrow[d,"\td^5"] \\
 E^5_{3,4}(Q) \cong  \tilde H_3(Q, \Omega^{\rm SO}_4(\rm pt)) \arrow[r,"p_*"'] & E^5_{3,5}(K(\Z,3)) \cong \tilde H_3(K(\Z,3),\Omega^{\rm SO}_4(\rm pt)) \cong \Z \,.
\end{tikzcd}
\end{equation}

We already observed in~\eqref{eq:d5KZ3} that the left-hand differential $\td^5$ is trivial for group value reasons, and that the lower map $p_*$ is an isomorphism by construction. By commutativity of the diagram, it follows that $\td_Q^5=0$. Thus, the only nontrivial $E^\infty$-terms in total degree $7$ are $E^\infty_{7,0}(Q) = \tilde H_7(Q,\Z) \cong \Z \oplus \Z_3$ and $E^\infty_{3,4}(Q) = \tilde H_3(Q,\Z) \cong \Z $.
The extension problem involving the $\Z_3$-summand of $E^\infty_{7,0}(Q)$ and the $\Z$-summand of $E^\infty_{3,4}(Q)$ was resolved in Theorem~\ref{theoremOmega7}. Therefore, we obtain
\begin{equation}
    \boxed{ \tilde \Omega^{\rm SO}_{7}(Q)\cong \Z^2}\,.
\end{equation}
These two free summands are detected by the bordism invariants $\iota\smile \gamma$ and $\iota\smile p_1$, respectively.

\paragraph{Spin case.} For spin bordism computation, we have Table~\ref{tab:E2spinQ} for relevant nonvanishing entries.
\begin{sseqdata}[ name = E2spinQ, classes = { draw = none }, axes type =  frame, scale = 0.8 ]

\class(0,0)
\class(0,7)

\class["\Z"](3,4)
\class["\Z_2"](3,2)

\class["\Z_2"](4,1)
\class["\Z_2"](4,2)

\class["\Z_2"](5,1)
\class["\Z_2"](5,2)

\class["\Z_2"](6,0)
\class["\Z^2_2"](6,1)
\class["\Z^2_2"](6,2)

\class["\Z^2_2"](7,1)

\class["C"](7,0)
\class["D"](8,0)

\end{sseqdata}
\begin{table}[htbp]
  \centering
  \printpage[
    name = E2spinQ,
    grid = chess,
  ]
  \caption{The relevant $E^2$-page entries for $E^2_{p,q} = H_p(Q,\Omega^{\rm Spin}_q({\rm pt}))$.}
  \label{tab:E2spinQ}
\end{table}
The differential
\begin{equation}
   \left(E^2_{6,0}(Q)\right)_{=\tilde H_6(Q,\Z) \cong \Z_2}  \longrightarrow   \left(E^2_{4,1}(Q) \right)_{= \tilde H_4(Q,\Z_2) \cong \Z_2}
\end{equation}
is the dual of mod-2 reduction composed with $\mathrm{Sq}^2$ and it is an isomorphism. Likewise,
\begin{equation}
\left( E^2_{5,1}(Q) \right)_{= \tilde H_5(Q,\Z_2) \cong \Z_2}  \longrightarrow   \left( E^2_{3,2}(Q)\right)_{= \tilde H_3(Q,\Z_2) \cong \Z_2}
\end{equation}
is the dual of $\mathrm{Sq}^2$ which is also an isomorphism.
The differential
\begin{equation}
  \left( E^2_{6,1}(Q) \right)_{ = \tilde H_6(Q,\Z_2) \cong \Z^2_2}
  \longrightarrow  \left(E^2_{4,2}(Q)\right)_{ = \tilde H_4(Q,\Z_2) \cong \Z_2 }
\end{equation}
is surjective by considering $\mathrm{Sq}^2$ acting on $\bar \gamma \in H^4(K(\Z,4),\Z_2)$.  The entire total-degree-6 diagonal is eliminated, and we can claim
\begin{equation}
   \boxed{ \tilde \Omega^{\rm Spin}_6(Q) = 0}\,.
\end{equation}

Now for the $p+q=7$ diagonal, the differential
\begin{equation}
  \left(  E^2_{7,0}(Q) \right)_{= \tilde H_7(Q,\Z) \cong \Z \oplus \Z_3} \longrightarrow \left( E^2_{5,1} (Q)\right)_{=  \tilde H_5(Q,\Z_2) \cong \Z_2 }
\end{equation} vanishes, since $\mathrm{Sq}^2$ vanishes on the dual cohomology groups. We also recall the computation~\eqref{eq:triviald280} and it follows that
$E^2_{8,0}(Q)  \longrightarrow  E^2_{6,1}(Q)$ is a trivial differential for the same reason. These lead to $E^\infty_{7,0}(Q) \cong \Z \oplus \Z_3$ and $E^\infty_{6,1}(Q) \cong \Z_2$.

New ingredients, beyond the single-$K(\Z,3)$ AHSS discussion, appear for the differential
\begin{equation}
    \td^2 : \left( E^2_{7,1}(Q) \right)_{= \tilde H_7(Q,\Z_2)\cong \Z_2^2 } \longrightarrow \left( E^2_{5,2}(Q)\right)_{= \tilde H_5(Q,\Z_2)\cong \Z_2} \,.
\end{equation}
Dually, we have $\tilde H^5(Q,\Z_2) \cong \Z_2 \,\mathrm{Sq}^2 u $ and  $\tilde H^7(Q,\Z_2) \cong \Z_2 u \bar \gamma \oplus \Z_2 \,\mathrm{Sq}^3 \bar \gamma$. Using the Adem relation and degree reasons
\begin{equation}
    \mathrm{Sq}^2(\mathrm{Sq}^2 u)=\mathrm{Sq}^3\mathrm{Sq}^1 u=0 \,,
\end{equation}
and the above $\td^2$ vanishes. Furthermore, the possible differential
\begin{equation}
    \td^3 : \left(E^3_{8,0}(Q) \right)_{= \tilde H_8(Q,\Z) \cong \Z \oplus \Z_3 \oplus \Z_2 } \longrightarrow  \left(E^3_{5,2}(Q)\right)_{\cong \Z_2} \,,
\end{equation}
is nothing but the same $K(\Z,3)$ $\td^3$-differential in \eqref{eq:triviald380}, which turns out to be trivial. Finally, it remains to consider the differential
\begin{equation}
    \td^5 :\left( E^5_{8,0}(Q) \right)_{ = \tilde H_8(Q,\Z) \cong \Z \oplus \Z_3 \oplus \Z_2}\longrightarrow \left( E^5_{3,4}(Q) \right)_{= \tilde H_3(Q,\Z)\cong \Z }\,.
\end{equation}
From the group structure alone, this differential could a priori be nontrivial. However, using the same projection $p$ as in the commutative diagram~\eqref{eq:commdiaQ}, we compare it with the corresponding differential for $K(\Z,3)$, namely~\eqref{eq:d5KZ3}. By the same naturality argument as before, it follows that this $\td^5$ also vanishes. The other final page entries are $E^\infty_{5,2}(Q) \cong \Z_2$ and $E^\infty_{3,4}(Q) \cong \Z$.

Since the $\Z$-summand in $E^\infty_{7,0}(Q) \cong \Z \oplus \Z_3$ is generated by $\iota \smile \gamma$, the filtration of the pure $K(\Z,3)$-part of $Q$ remains the same as in \eqref{eq:filtrationSpinKZ}, we get the result
\begin{equation}
     \boxed{ \tilde \Omega^{\rm Spin}_{7}(Q)\cong \Z^2}\,,
\end{equation}
with the dual basis $\iota\smile \gamma$ and $\frac{1}{4} \iota\smile p_1$.

\section{Computation of $(I_{\Z} \tilde \Omega^{\rm SO/Spin})^{9}(K(\Z,3) \times K(\Z,6))$}
\label{sec:OmegaKZ3KZ6}

Let $T = K(\Z,3) \times K(\Z,6)$ denote the classifying space under consideration. From the general discussion of invertible phases in dimension $d+2=7+2=9$, our goal is to determine the torsion subgroup of $\tilde \Omega^{\rm SO/Spin}_8(T)$ and the free part of $\tilde \Omega^{\rm SO/Spin}_9(T)$.

To this end, it is economical to compute the full group $\tilde \Omega^{\rm SO/Spin}_8(T)$, while extracting only the free part of $\tilde \Omega^{\rm SO/Spin}_9(T)$.
For the reduced (co)homology of $K(\Z,6)$, we have the classical results~\cite{Cartan1954EM}:
\begin{equation}
\tilde H^i(K(\Z,6),\Z)=
\begin{cases}
\Z\, \zeta & i=6\\
\Z_2\, \beta \mathrm{Sq}^2 \bar \zeta & i=9\\
0 & i=1,2,3,4,5,7,8
\end{cases}\,,
\end{equation}
and
\begin{equation}
\tilde H_i(K(\Z,6),\Z)=
\begin{cases}
\Z & i=6\\
\Z_2 & i=8\\
0& i=1,2,3,4,5,7,9
\end{cases}\,,
\end{equation}
where we use $\zeta$ to denote the generator of $\tilde H^6 (K(\Z,6),\Z)$ and
$\beta: H^8(K(\Z,6),\Z_2) \rightarrow  H^{9}(K(\Z,6),\Z)$ is the Bockstein homomorphism for $0 \rightarrow \Z \xrightarrow{\times \ 2} \Z \rightarrow \Z_2\rightarrow 0$.

The mod-2 results are
\begin{equation}
\tilde H^i\!\left(K(\Z,6),\Z_2\right)=
\begin{cases}
\Z_2 \,\bar \zeta & i=6\\
\Z_2 \, \mathrm{Sq}^2 \bar \zeta &i=8\\
\Z_2 \, \mathrm{Sq}^3 \bar \zeta &i=9\\
0 & i=1,2,3,4,5,7
\end{cases}\,,
\end{equation}
and
\begin{equation}
    \tilde H_i(K(\Z,6),\Z_2)=
\begin{cases}
\Z_2 & i=6,8,9\\
0 & i=1,2,3,4,5,7
\end{cases}\,,
\end{equation}
with generators $\bar \zeta\in H^6(K(\Z,6),\Z_2)$ and an identity that the mod-2 reduction of $\beta \mathrm{Sq}^2 \bar \zeta$ equals $\mathrm{Sq}^3 \bar \zeta$ by Adem relations.

\begin{table}[H]
\centering
    \begin{tabular}{c|ccccccc}
  $p$ &  $3$ & $4$ & $5$ & $6$ & $7$ & $8$ & $9$ \\ \midrule
$\tilde H^p(T,\Z)$  & $\Z \iota$ & $ 0 $ & $0$ & $\Z\,\zeta \oplus\Z_2 \,\iota^2$ & $0$ & $  \Z_3\, \kappa $ &$  \Z \,\iota \zeta \oplus \Z_2 \,\iota^3 \oplus \Z_2 \, \beta \mathrm{Sq}^2 \bar \zeta  $ \\ \midrule
$\tilde H_p(T,\Z)$   & $\Z$ & $ 0 $ & $\Z_2$ & $ \Z $ & $\Z_3$ & $ \Z_2 \oplus \Z_2 $ &$\Z \oplus \Z_2$ \\ \midrule
$\tilde H^p(T,\Z_2)$ & $\Z_2 \,u $ & $0 $ & $\Z_2 \, \mathrm{Sq}^2 u$ & $\Z_2 \, \bar \zeta  \oplus \Z_2\,u^2 \,  $ & $ 0$ & $\Z_2 \, u\,\mathrm{Sq}^2 u\oplus \Z_2 \,\mathrm{Sq}^2 \bar \zeta$ & $\Z_2 \, u\bar \zeta \oplus \Z_2 \, u^3 \oplus \Z_2\, \mathrm{Sq}^3 \bar \zeta$\\ \midrule
$\tilde H_p(T,\Z_2)$  & $\Z_2$ & $ 0 $ & $\Z_2$ & $\Z_2^2$ & $0$ & $\Z_2^2$ & $\Z_2^3$
\end{tabular}
\caption{(Co)homology data of $T=K(\Z,3) \times K(\Z,6)$.}
\end{table}

\paragraph{Oriented case.}
The relevant nonvanishing entries are presented in Table~\ref{tab:E2orientedT}. We do not show the $p+q=10$ diagonal, since any differential from it does not affect the free part of $\tilde \Omega^{\rm SO}_9(T)$ we are interested in, which comes only from the $q=0$ row. Here $F = \Z \oplus \Z_2$.
\begin{sseqdata}[ name = E2orientedT, classes = { draw = none }, axes type =  frame, scale = 0.8 ]
\class(0,0)
\class(0,9)
\class["\Z"](3,4)
\class["\Z_2"](3,5)
\class["\Z_2"](5,4)
\class["\Z_3"](7,0)
\class["\Z^2_2"](8,0)
\class["F"](9,0)
\end{sseqdata}
\begin{table}[htbp]
  \centering
  \printpage[
    name = E2orientedT,
    grid = chess,
  ]
  \caption{The relevant $E^2$-page entries for $E^2_{p,q} = H_p(T,\Omega^{\rm SO}_q({\rm pt}))$.}
  \label{tab:E2orientedT}
\end{table}
Due to the values of the groups involved, there are no nontrivial differentials from the diagonal $p+q=8$ to the diagonal $p+q=7$. Consider the projection $p' : T = K(\Z,3) \times K(\Z,6) \rightarrow K(\Z,3)$. Using the induced map on the AHSS, we can relate the potentially nontrivial differentials
\begin{equation}
\left( E^2_{5,4}(T) \right)_{= \tilde H_5(T,\Z) \cong \Z_2}  \longrightarrow   \left( E^2_{3,5}(T)\right)_{= \tilde H_3(T,\Z_2) \cong \Z_2}\,,
\end{equation}
and
\begin{equation}
\left( E^6_{9,0}(T) \right)_{= \tilde H_9(T,\Z) \cong \Z\oplus\Z_2}  \longrightarrow   \left( E^6_{3,5}(T)\right)_{= \tilde H_3(T,\Z_2) \cong \Z_2}\,,
\end{equation}
to the $\td^2$ and $\td^6$ in Lemma~\ref{lemma1}. By naturality, both of the above differentials vanish, exactly as in~\eqref{eq:commdiaQ}. This leads to the following result 
\begin{equation}
    \boxed{\tilde \Omega^{\rm SO}_8(T) \cong   \Z^3_2}\qquad \text{and} \qquad
    \boxed{ \text{Free} \ \tilde \Omega^{\rm SO}_9(T) \cong   \Z}\,.
\end{equation}
The only possible remaining extension would be an extension of
$E^\infty_{8,0}(T)=(\Z_2)^2$ by $E^\infty_{3,5}(T)=\Z_2$. This extension is
trivial. Indeed, the summand $E^\infty_{3,5}(T)$ comes entirely from the
$K(\Z,3)$ factor of $T=K(\Z,3)\times K(\Z,6)$, and
Theorem~\ref{theorem:3.3} shows that the corresponding $\Z_2$-summands from
$K(\Z,3)$ do not combine into a nontrivial extension. The remaining $\Z_2$
summand in $E^\infty_{8,0}(T)$ comes from the $K(\Z,6)$ factor, which is
independent in the cartesian product $T$. Hence the filtration splits as $(\Z_2)^3$.

Hence, together with our previous result for $K(\Z,3)$ we have
\begin{equation}
    \boxed{(I_{\Z} \tilde \Omega^{\rm SO})^{9}(K(\Z,3) \times K(\Z,6)) \cong \Z \,\iota \zeta \oplus \Z_2 \, u\,\mathrm{Sq}^2 u \oplus \Z_2 \,u \ w_2 \ w_3 \oplus \Z_2 \,\bar \zeta w_2 }  \,.
\end{equation}

\paragraph{Spin case.} Table~\ref{tab:E2spinT} contains relevant nonvanishing entries for the spin bordism computation.
\begin{sseqdata}[ name = E2spinT, classes = { draw = none }, axes type =  frame, scale = 0.8 ]

\class(0,0)
\class(0,7)
\class(0,9)

\class["\Z"](3,4)

\class["\Z_2"](5,2)
\class["\Z_2"](5,4)

\class["\Z^2_2"](6,1)
\class["\Z^2_2"](6,2)

\class["\Z_3"](7,0)

\class["\Z^2_2"](8,0)
\class["\Z^2_2"](8,1)
\class["F"](9,0)

\end{sseqdata}
\begin{table}[htbp]
  \centering
  \printpage[
    name = E2spinT,
    grid = chess,
  ]
  \caption{The relevant $E^2$-page entries for $E^2_{p,q} = H_p(T,\Omega^{\rm Spin}_q({\rm pt}))$.}
  \label{tab:E2spinT}
\end{table}
The differential
\begin{equation}
\left( E^2_{8,0}(T) \right)_{= \tilde H_8(T,\Z) \cong \Z^2_2}  \longrightarrow   \left( E^2_{6,1}(T)\right)_{= \tilde H_6(T,\Z_2) \cong \Z^2_2},
\end{equation}
is nontrivial. Cohomologically, it acts on the generators as
\begin{equation}
    \begin{aligned} \label{eq:80618162}
       u^2 &\longmapsto \mathrm{Sq}^2 u^2 = 0 \mod 2 \\
       \bar \zeta &\longmapsto \mathrm{Sq}^2 \bar \zeta \neq 0\,,
    \end{aligned}
\end{equation}
where $\mathrm{Sq}^2 \bar \zeta$ generates a $\Z_2$-summand of $\tilde H_6(T,\Z_2)$. It follows that $E^3_{8,0}(T) \cong E^3_{8,0}(K(\Z,3))\cong \Z_2$.
The differential on the $E^3$-page
\begin{equation}
\left( E^3_{8,0}(T) \right)_{\cong \Z_2}  \longrightarrow   \left( E^3_{5,2}(T)\right)_{= \tilde H_5(T,\Z_2) \cong \Z_2},
\end{equation} matches exactly the differential~\eqref{eq:triviald380} in the $K(\Z,3)$ case, which is trivial.
The differential
\begin{equation}
\left( E^2_{8,1}(T) \right)_{= \tilde H_8(T,\Z_2) \cong \Z^2_2}  \longrightarrow   \left( E^2_{6,2}(T)\right)_{= \tilde H_6(T,\Z_2) \cong \Z^2_2}
\end{equation}
follows exactly the same pattern as~\eqref{eq:80618162}, and hence $E^3_{6,2}(T) \cong \Z_2$.
Now on the $E^3$-page, the differential
\begin{equation}
\left( E^3_{9,0}(T) \right)_{\cong \Z \oplus \Z_2}  \longrightarrow   \left( E^3_{6,2}(T)\right)_{ \cong \Z_2}
\end{equation}
is surjective and the $\Z_2$-summand of $E^3_{9,0}(T)$ is mapped isomorphically to $E^3_{6,2}(T)$, same as in~\eqref{eq:isod390}.
We conclude that $E^\infty_{8,0} \cong \Z_2$ and $E^\infty_{9,0} \cong \Z$, such that \begin{equation}
    \boxed{\tilde \Omega^{\rm Spin}_8(T) \cong   \Z_2}\qquad \text{and} \qquad
    \boxed{ \text{Free} \, \tilde \Omega^{\rm Spin}_9(T) \cong   \Z}\,.
\end{equation}
\begin{equation}
    \boxed{(I_{\Z} \tilde \Omega^{\rm Spin})^{9}(K(\Z,3) \times K(\Z,6)) \cong \Z \,\iota \zeta \oplus \Z_2 \, u\,\mathrm{Sq}^2 u }  \,.
\end{equation}

\providecommand{\href}[2]{#2}\begingroup\raggedright\endgroup

\end{document}